\documentclass[10pt, conference, letterpaper]{IEEEtran}
\IEEEoverridecommandlockouts

\usepackage{algorithmic,algorithm}
\usepackage[utf8]{inputenc} 
\usepackage{url}            
\usepackage{booktabs}       
\usepackage{amsfonts}       
\usepackage{nicefrac}       
\usepackage{microtype}      

\usepackage{dsfont}
\usepackage{multirow}
\usepackage{graphicx}
\usepackage{color,soul}
\usepackage{amsmath,amssymb,amsfonts,amsthm}
\usepackage{enumerate}
\usepackage{epstopdf}
\graphicspath{ {images/} }
\usepackage{tikz}
\usepackage{mathtools}
\usepackage{tabulary}
\usepackage{algorithm}
\usepackage{caption}

\usepackage{xspace}
\usepackage{subfigure}
\usepackage{cleveref}
\newtheorem{theorem}{Theorem}

\newtheorem{lemma}{Lemma}
\newtheorem{remark}{Remark}
\newtheorem{proposition}{Proposition}

\newtheorem{assumption}{Assumption}

\newcommand{\cA}{\mathcal{A}}

\newcommand{\cS}{\mathcal{S}}

\newtheorem{definition}{Definition}

\newcommand{\qw}{\texttt{Q-Whittle}\xspace}

\newcommand{\qwhittle}{\texttt{Q$^+$-Whittle}\xspace}
\newcommand{\qwhittleLFA}{\texttt{Q$^+$-Whittle-LFA}\xspace}

\begin{document}


\title{Whittle Index based Q-Learning for Wireless Edge Caching with Linear Function Approximation}
\author{Guojun~Xiong,~\IEEEmembership{Student Member,~IEEE,}
            Shufan~Wang,
            Jian~Li,~\IEEEmembership{Member,~IEEE,} and~Rahul Singh,~\IEEEmembership{Member,~IEEE} 
\IEEEcompsocitemizethanks{\IEEEcompsocthanksitem G. Xiong, S. Wang,  and J. Li are with Binghamton University, State University of New York, Binghamton, NY, 13902.
E-mail: \{gxiong1, swang214, lij\}@binghamton.edu.}
\IEEEcompsocitemizethanks{\IEEEcompsocthanksitem R. Singh is with Indian Institute of Science. Email: rahulsingh@iisc.ac.in.}
}

\maketitle

\begin{abstract}
We consider the problem of content caching at the wireless edge to serve a set of end users via unreliable wireless channels so as to minimize the average latency experienced by end users due to the constrained wireless edge cache capacity. We formulate this problem as a Markov decision process, or more specifically a restless multi-armed bandit problem, which is provably hard to solve.  We begin by investigating a discounted counterpart, and prove that it admits an optimal policy of the threshold-type.  We then show that this result also holds for average latency problem.  Using this structural result, we establish the indexability of our problem, and employ the Whittle index policy to minimize average latency. Since system parameters such as content request rates and wireless channel conditions are often unknown and time-varying, we further develop a model-free reinforcement learning algorithm  dubbed as \qwhittle that relies on Whittle index policy. However, \qwhittle requires to store the Q-function values for all state-action pairs, the number of which can be extremely large for wireless edge caching. To this end, we approximate the Q-function by a parameterized function class with a much smaller dimension, and further design a \qwhittle algorithm with linear function approximation, which is called \qwhittleLFA.  We provide a finite-time bound on the mean-square error of \qwhittleLFA.  Simulation results using real traces demonstrate that \qwhittleLFA yields excellent empirical performance.

\end{abstract}

\begin{IEEEkeywords}
Wireless Edge Caching, Restless Bandits, Whittle Index Policy, Reinforcement Learning, Finite-Time Analysis
\end{IEEEkeywords}


\section{Introduction}\label{sec:intro}

\IEEEPARstart{T}{he} dramatic growth of wireless traffic due to an enormous increase in the number of mobile devices is posing many challenges to the current mobile network infrastructures.~In addition to this increase in the volume of traffic, many emerging applications such as Augmented/Virtual Reality, autonomous vehicles and video streaming, are \textit{latency-sensitive}. In view of this, the traditional approach of offloading the tasks to remote data centers is becoming less attractive. Furthermore, since these emerging applications typically require unprecedented computational power, it is not possible to run them on mobile devices, which are typically resource-constrained.

To provide such stringent timeliness guarantees, mobile edge computing architectures have been proposed as a means to improve the quality of experience (QoE) of end users, which move servers from the cloud to edges, often wirelessly that are closer to end users.  Such edge servers are often empowered with a small wireless base station, e.g., the storage-assisted future mobile Internet architecture and cache-assisted 5G systems \cite{andrews2012femtocells}.
By using such edge servers, content providers are able to ensure that contents such as movies, videos, software, or services are provided with a high QoE (with minimal latency). The success of edge servers relies upon ``content caching'', for which popular contents are placed at the cache associated with the wireless edge.  If the content requested by end users is available at the wireless edge, then it is promptly delivered to them. Unfortunately, the amount of contents that can be cached at the wireless edge is often limited by the wireless edge cache capacity.  These issues are further exacerbated when the requested content is delivered over \emph{unreliable} channels.

In this work, we are interested in minimizing \textit{the average latency} incurred while delivering contents to end users, which are connected to a wireless edge via unreliable channels. We design dynamic policies that decide \textit{which contents should be cached at the wireless edge so as to minimize the average latency of end users.}

\subsection{Whittle Index Policy for Wireless Edge Caching} We pose this problem as a Markov decision process (MDP) \cite{puterman1994markov} in Section~\ref{sec:model}. Here, the system state is the number of outstanding requests from end users for each content that needs to be satisfied, and the cost is measured as the latency experienced by end users to obtain the requested contents.  The available actions are the choices of caching each content or not given that the wireless edge cache capacity is much smaller than the total number of distinct requested contents.
This MDP turns out to be an infinite-horizon average-cost restless multi-armed bandit (RMAB) problem \cite{whittle1988restless}. Even though in theory this RMAB can be solved by using relative value iteration \cite{puterman1994markov}, this approach suffers from {the curse of dimensionality}, and fails to provide any insight into the solution. Thus, it is desirable to derive low-complexity solutions and provide guarantees on their performance.  A celebrated policy for RMAB is \textit{the Whittle index policy} \cite{whittle1988restless}.  We propose to use the Whittle index policy to solving the above problem for wireless edge caching.

Following the approach taken by Whittle \cite{whittle1988restless}, we begin by relaxing the hard constraint of the original MDP, which requires that the number of cached contents at each time is \textit{exactly} equal to the cache size. These are relaxed to a constraint which requires that the number of cached contents is equal to the cache size \textit{on average}. We then consider the Lagrangian of this relaxed problem, which yields us a set of decoupled average-cost MDPs, which we call the per-content MDP. Instead of optimizing the average cost (latency) of this per-content MDP, we firstly consider a discounted per-content MDP, and prove that the optimal policy for each discounted per-content MDP has an appealing \textit{threshold structure}. This structural result is then shown to also hold for \textit{the average latency problem}. We use this structural result to show that our problem is indexable~\cite{whittle1988restless}, and then derive Whittle indices for each content.  Whittle index policy then prioritizes contents in a decreasing order of their Whittle indices, and caches the maximum number of content constrained by the cache size. Whittle index policy is computationally tractable since its complexity increases linearly with the number of contents.  Moreover it is known to be asymptotically optimal \cite{weber1990index,verloop2016asymptotically} as the number of contents and the cache size are scaled up, while keeping their ratio as a constant. Our contribution in Section~\ref{sec:index} is non-trivial since establishing indexability of RMAB problems is typically intractable in many scenarios, especially when the probability transition kernel of the MDP is convoluted \cite{nino2007dynamic}, and Whittle indices of many practical problems remain unknown except for a few special cases.

\subsection{Whittle Index-based Q-learning with Linear Function Approximation for Wireless Edge Caching}
The Whittle index policy needs to know the controlled transition probabilities of the underlying MDPs, which in our case amounts to knowing the statistics of content request process associated with end users, as well as the reliability of wireless channels.  However, these parameters are often unknown and time-varying.  Hence in Section~\ref{sec:learning}, we design an efficient reinforcement learning (RL) algorithm to make optimal content caching decisions dynamically without knowing these parameters.  We do not directly apply off-the-shelf RL methods such as UCRL2~\cite{jaksch2010near} and Thompson Sampling~\cite{gopalan2015thompson} since the size of state-space grows exponentially with the number of contents, and hence the computational complexity and the learning regret would also grow exponentially. Thus, the resulting algorithms would be too slow to be of any practical use.

To overcome these limitations, we first derive a model-free RL algorithm dubbed as \qwhittle, which is largely inspired by the recent work \cite{avrachenkov2020whittle} that proposed a tabular Whittle index-based Q-learning algorithm, which we call \qw for ease of exposition. The key aspect of \qw \cite{avrachenkov2020whittle} is that the updates of Q-function values and Whittle indices form a two-timescale stochastic approximation (2TSA) \cite{borkar2009stochastic,konda2000actor} with the former operating on a faster timescale and the latter on a slower timescale.  Though \cite{avrachenkov2020whittle} provided a rigorous asymptotic convergence analysis, such a 2TSA usually suffers from slow convergence in practice (as we numerically verify in Section~\ref{sec:evaluation}).  To address this limitation, our key insight is that we can further leverage the threshold-structure of the optimal policy to each per-content MDP to learn Q-function values of \textit{only} those state-action pairs which are visited under the current threshold policy, rather than all state-action pairs as in \qw.  This novel update rule enables \qwhittle to significantly improve the sample efficiency of \qw using the conventional $\epsilon$-greedy policy.

We note that \qwhittle needs to store Q-function values for all state-action pairs, the number of which can be very large for wireless edge caching.  To address this difficulty, we further study \qwhittle with linear function approximation (LFA) by using low-dimensional linear approximation of Q-function.  We call this algorithm the \qwhittleLFA, which can be viewed through the lens of a 2TSA.  We provide a finite-time bound on the mean-square error of \qwhittleLFA in Section~\ref{sec:finite-time-analysis}.  To the best of our knowledge, our work is the first to consider a model-free RL  approach with LFA towards a Whittle index policy in the context of wireless edge caching over unreliable channels, and the first to provide a finite-time analysis of a Whittle index based Q-learning with LFA.

Finally, we provide extensive numerical results using both synthetic and real traces to support our theoretical findings in Section~\ref{sec:evaluation}, which demonstrate that \qwhittleLFA produces significant performance gain over state of the arts.

\section{Related Work}\label{sec:related}

Although edge caching has received a significant amount of attentions, we are not aware of any prior work proposing an analytical model for latency-optimal wireless edge caching over unreliable channels, designing a computationally efficient index based policy and a novel RL augmented algorithm in face to unpredictable content requests and unreliable channels.  We provide an account of existing works in two areas closely related to our work: content caching and restless bandits.

\textbf{Content Caching} \cite{paschos2018role} has been studied in numerous domains with different objectives such as minimizing expected delay \cite{poularakis2018distributed,vu2018latency}, operational costs \cite{abolhassani2020achieving} or
maximizing utility \cite{dehghan2019utility}. The joint caching and request routing has also been investigated, e.g., \cite{ioannidis2016adaptive,li2018dr,mahdian2019kelly}. Most prior works formulated the problem as a constrained/stochastic optimization problem, etc.  None of those works provided a formulation using the RMAB framework and developed an index based caching policy.  Furthermore, all above works assumed full knowledge of the content request processes and hence did not incorporate a learning component. A recent line of works considered  caching from an online learning perspective, e.g., \cite{garg2019online,wei2021wireless,zhao2018red,paschos2019learning,bhattacharjee2020fundamental,salem2021no},
and used the performance metric of learning regret or competitive ratio. Works such as \cite{sadeghi2017optimal,sadeghi2019deep,somuyiwa2018reinforcement,wang2020intelligent}
used deep RL methods.  However, deep RL methods  lack of theoretical performance guarantees.  Our model, objective and formulation significantly depart from those considered in aforementioned works, where we pose the wireless edge caching problem as a MDP and develop the Whittle index policy that can be easily learned through a model-free RL framework.

\textbf{Restless Multi-Armed Bandit} (RMAB) is a general framework for sequential decision making problems, e.g., \cite{larrnaaga2016dynamic,larranaga2014index}.  However, RMAB is notoriously intractable \cite{papadimitriou1994complexity}. One celebrated policy is the Whittle index policy \cite{whittle1988restless}.
However, Whittle index is well-defined only when the {indexability condition} is satisfied, which is in general hard to verify.  Additionally, the application of  Whittle index requires full system knowledge.
 Thus it is important to examine RMAB from a learning perspective, e.g., \cite{dai2011non,liu2012learning,tekin2012online,ortner2012regret,jung2019regret}.
 However, these methods did not exploit the special structure available in RMAB and contended directly with an extremely high dimensional state-action space yielding the algorithms to be too slow to be useful.  Recently,  RL based algorithms have been developed \cite{avrachenkov2020whittle,fu2019towards,wang2020restless,robledo2022qwi,robledo2022tabular,xiong2022reinforcement,xiong2022index,nakhleh2022deeptop,xiong2023reinforcement}
to explore the problem structure through index policies.  However, \cite{avrachenkov2020whittle,fu2019towards,robledo2022qwi,robledo2022tabular} lacked finite-time performance analysis and multi-timescale SA algorithms often suffer from slow convergence.   \cite{wang2020restless,xiong2022reinforcement} depended on a simulator for explorations which cannot be directly applied here since it is difficult to build a perfect simulator in complex wireless edge environments. \cite{nakhleh2022deeptop} leveraged the threshold policy via a deep neural network without finite-time performance guarantees. \cite{xiong2022index,xiong2023reinforcement} either studied a finite-horizon setting or developed model-based RL solutions, while we consider an infinite-horizon average-cost setting and develop model-free RL algorithms.   Specifically, we propose \qwhittleLFA, a low-complexity Whittle index based Q-learning algorithm with linear function approximation.  Our finite-time analysis of \qwhittleLFA further distinguishes our work.

\begin{figure}
	\centering
	\includegraphics[width=0.3\textwidth]{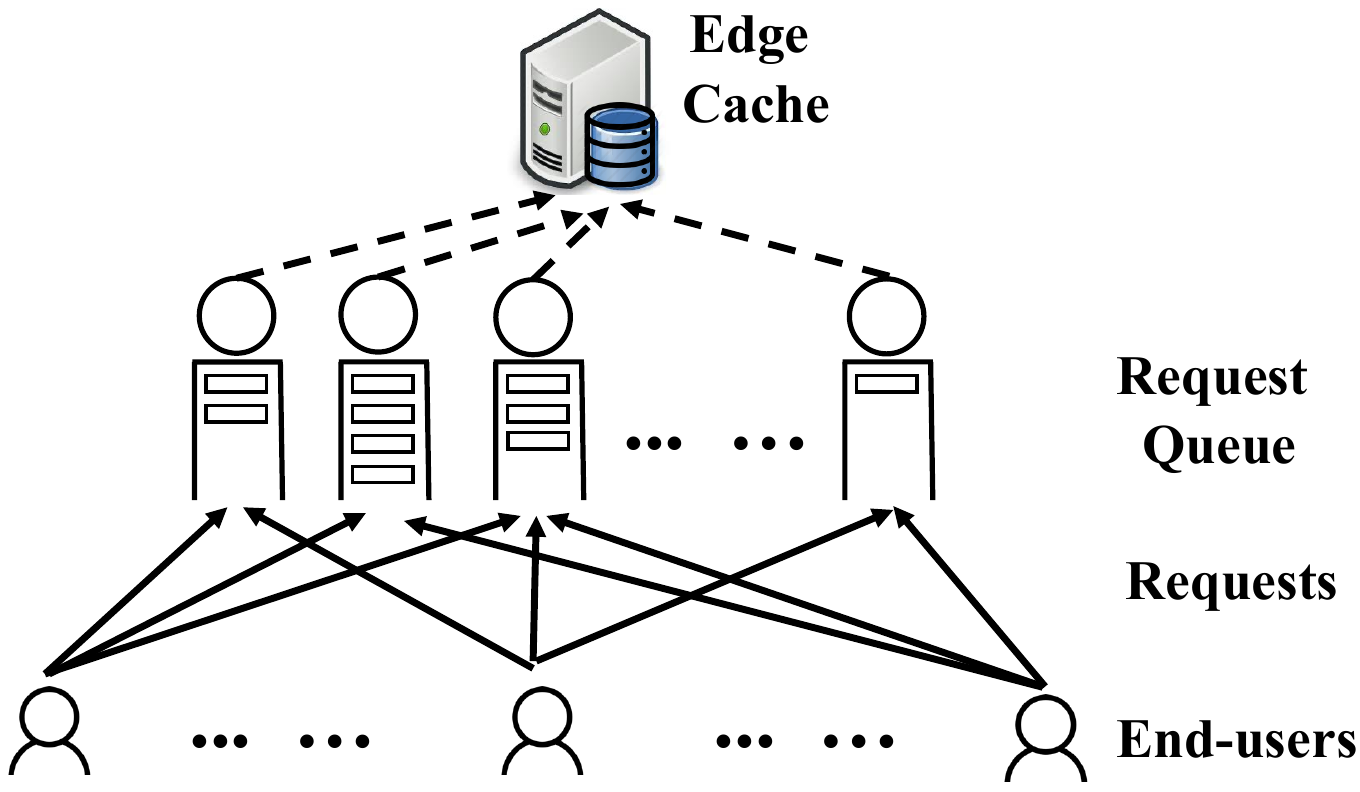}
	\vspace{-0.05in}
	\caption{Wireless edge caching over unreliable channels.}
	\label{fig:model}
	\vspace{-0.1in}
\end{figure}

\section{System Model and Problem Formulation}\label{sec:model}
In this section, we present the system model and formulate the average latency minimization problem for wireless edge caching over unreliable channels.

\subsection{System Model}

Consider a wireless edge system as shown in Figure~\ref{fig:model}, where the wireless edge is equipped with a cache of size $B$ units to store contents that are provided to end users. We denote the set of distinct contents as $\mathcal{M}=\{1, \cdots, M\}$ with $|\mathcal{M}|=M$. We assume that all contents are of unit size. End users make requests for different contents to the wireless edge. If the requested content is available at the wireless edge, then it is delivered to end users directly through a wireless channel that is unreliable \cite{wei2021wireless}. The goal of the wireless edge is to decide at each time which contents to cache so that the cumulative value of the average content request latency experienced by end users is minimal.

\textbf{Content Request and Delivery Model.} We assume that requests for content $m\in\mathcal{M}$ arrive at the wireless edge from end users according to a Poisson process\footnote{Poisson arrivals have been widely used in the literature, e.g., \cite{ioannidis2016adaptive,li2018dr,wei2021wireless} and references therein.  However, our model holds for general stationary process  \cite{baccelli13}  and our model-free RL based algorithm and analysis in Sections~\ref{sec:learning} and~\ref{sec:finite-time-analysis} holds for any request process.} with arrival rate $\lambda_m$.   To each content $m$, we associate a ``request queue'' at the wireless edge, which stores the number of outstanding requests for content $m$ at time $t$. The queue length associated with the number of such requests at time $t$ is denoted by $S_{m, t}$. The rationality of this model is that the number of content requests may be larger than the service capacity of the wireless edge server \cite{atre2020caching,xiong2023reinforcement}.  Hence, the content requested from an end user might not be served immediately so that there will be a latency associated with the end user getting content.  Another point is due to the fact that the wireless channels between the edge cache and end users are unreliable. This motivates us to consider a queuing model which captures the latency experienced by end users.

The time taken by the wireless edge  to deliver the content, i.e. serve end users' requests, is modeled by appropriate random variables, which heavily relies on the wireless channel quality between the wireless edge and end users\footnote{Though the wireless channel can be explicitly modeled as in physical-layer communication models \cite{vu2018latency}, it requires additional beamforming and channel estimations, which is out of the scope of this work.  With our queue model, the effect of wireless channel is incorporated in content departure rate as in \cite{wei2021wireless}. }.
More specifically, we assume that the time taken to deliver content $m$ to end users is exponentially distributed with mean $1\slash \nu_m$ \cite{wei2021wireless}. The service times are independently across different contents and requests. Thus, when $S_{m,t}\geq 1$, the request of content $m$ departs from the corresponding request queue with rate $\nu_m$.

\textbf{Decision Epochs.}  The decision epochs\slash times are the moments when the states of request queues change. At each decision epoch\slash time $t$, the wireless edge determines for each content whether or not it should be cached, and then delivers the cached contents to the desired end users.

\subsection{System Dynamics and Problem Formulation}\label{sec:mdp}
We now formulate the problem of average latency minimization for the above model as a MDP.

\textbf{States.} We denote the state of the wireless edge at time $t$ as $\bold{S}_t : =(S_{1,t}, \cdots, S_{M,t})\in\mathbb{N}^M$, where $S_{m,t}$ is the number of outstanding requests for content $m\in\mathcal{M}$. To guarantee the stability of the Markov chain, we assume that $S_{m,t}, \forall m, t$ is upper bounded by $S_{max}$, which can be arbitrarily large but bounded.  For ease of readability, we denote the state-space associated with $\bold{S}_t$ as $\mathcal{S}$.

\textbf{Actions.} At each time $t$, for each content $m$, the wireless edge has to make a decision regarding whether or not to  cache it.  We use $A_{m,t}$ to denote the action for content $m$ at time $t$. Thus, we let $A_{m,t}=1$ if it is cached, and $A_{m,t}=0$ otherwise.~We let $\mathcal{A}:=\{0,1\}$ be the set of decisions available for each content, and let $\bold{A}_{t} :=(A_{1,t}, \cdots, A_{M,t})$ be the vector consisting of decisions for $M$ contents.
The cache capacity constraint implies that $\bold{A}_{t}$ must satisfy the following constraints,
\begin{align}\label{eq:capacity-constraint}
 \sum_{m=1}^M A_{m,t}\leq B, \quad\forall t.
\end{align}
We aim to design a policy $\pi:\cS \mapsto \cA^{M}$ maps the state $\bold{S}_t$ of the wireless edge to caching decisions $\bold{A}_t$, i.e., $\bold{A}_t=\pi(\bold{S}_t)$.

\textbf{Transition Kernel.} The state of the $m$-th request queue can change from $S_m$ to either $S_m+1$, or $S_m-1$ at the beginning of each decision epoch. Let $\bold{e}_{m}$ be the $M$-dimensional vector whose $m$-th component is $1$, and all others are $0$. Then,
\begin{align}\label{eq:queuetrans}
    \bold{S} =
    \begin{cases}
    \bold{S} + \bold{e}_m, \ &\text{with transition rate} ~b_m(S_m, A_m), \\
    \bold{S}-\bold{e}_m, &\text{with transition rate} \ d_m(S_m, A_m),
    \end{cases}
\end{align}
where $b_{m}(S_m, A_m):=\lambda_m$. We allow for {state-dependent} content delivery rates, which enables us to model realistic settings  \cite{larranaga2014index,larrnaaga2016dynamic}.  In particular, our setup can cover the classic $M/M/k$ queue if $d_m(S_m,A_m)=\nu_mS_mA_m$. This models the general multicast scenario in which the wireless edge can simultaneously serve end users whose requested contents are cached at the edge.

\textbf{Average Latency Minimization Problem.}
It follows from Little's Law~\cite{john1961little} that the objective of minimizing the average latency faced by end users is equivalent to that of minimizing the average number of cumulative outstanding requests in the system. Let $C_{m,t}(S_{m,t},A_{m,t}):= S_{m,t}$ be the instantaneous cost incurred by user $m$ at time $t$, so that the cumulative cost incurred in the system at time $t$ is given by
\begin{align}
    C_t(\bold{S}_t, \bold{A}_t)= \sum_{m=1}^M C_{m,t}(S_{m,t},A_{m,t})= \sum_{m=1}^M {S_{m,t}}.
\end{align}
With this choice of instantaneous cost, the average cost incurred in the system is proportional to the average latency faced by end users. Our objective is to derive a policy $\pi$ that makes content caching decisions at the capacity-constrained wireless edge for solving the following MDP:
\begin{align}\label{eq:cmdp-opt1}
    \min_{\pi\in\Pi} ~&C_\pi:=\limsup_{T\rightarrow\infty}\sum_{m=1}^M \frac{1}{T}\mathbb{E}_\pi \left[\int_0^T S_{m,t}dt\right],\nonumber \displaybreak[0]\\
    \quad&\text{s.t.}~\sum_{m=1}^M A_{m,t}\leq B, ~\forall t,
\end{align}
where the subscript denotes the fact that the expectation is taken with respect to the measure induced by the policy $\pi,$ and $\Pi$ is the set of all feasible wireless edge caching policies. Henceforth, we refer to~(\ref{eq:cmdp-opt1}) as the ``original MDP.'' Since it is an infinite-horizon average-cost problem, in principle it can be solved via the relative value iteration  \cite{puterman1994markov}. More specifically, there exists a value function, and an average cost value for the above MDP \cite[Theorem 8.4.3]{puterman1994markov}:
\begin{lemma}
Consider the MDP~\eqref{eq:cmdp-opt1} whose transition kernel is described in (\ref{eq:queuetrans}). There exists a $\beta^{*}$ and a function $V: \cS\mapsto \mathbb{R}$ that satisfy the following \emph{dynamic programming (DP) equation}:
\begin{align}\label{eq:dynamicpe}
   \beta^* &= \min_{\sum_m A_{m}\leq B}\Bigg( \sum_{m=1}^M \Big[S_m + \lambda_m V(\bold{S} +\bold{ e}_m)\nonumber\allowdisplaybreaks\\
  & +\nu_mS_mA_m V(\bold{S}\!-\! \bold{e}_m)\! -\!(\lambda_m\!+\!\nu_mS_mA_m) V(\bold{S})\Big]\Bigg).
\end{align}
\end{lemma}
Though one can obtain an optimal policy $\pi^*$ using relative value iteration, this approach suffers from the curse of dimensionality, i.e., the computational complexity grows linearly with the size of state space $\cS$, the latter quantity in turn grows exponentially with the number of contents $M$. This renders such a solution impractical. Furthermore, this approach fails to provide insight into the solution structure. Thus, our focus will be on developing computationally appealing solutions.

\subsection{Lagrangian Relaxation}
We now discuss Lagrangian relaxation of the original MDP~(\ref{eq:cmdp-opt1}), and introduce the corresponding ``per-content MDP.''  The Lagrangian multipliers together with these per-content problems form the building block of our Whittle index policy, that will be formally introduced in Section~\ref{sec:index}.

Following Whittle's approach \cite{whittle1988restless}, we first consider the following \textit{``relaxed problem,''} which relaxes the ``hard'' constraint in~(\ref{eq:cmdp-opt1}) to an average constraint:
\begin{align}\label{eq:mdp-relaxed}
    \min_{\pi\in\Pi} ~&\limsup_{T\rightarrow\infty}\sum_{m=1}^M \frac{1}{T}\mathbb{E}_\pi \left[\int_{0}^T S_{m,t}dt\right], \nonumber \displaybreak[0]\\
    \quad&\text{s.t.}~ \limsup_{T\rightarrow\infty}\sum_{m=1}^M\frac{1}{T}\mathbb{E}_\pi\left[\int_{0}^T A_{m,t}dt\right]\leq B.
\end{align}
Next, we consider the Lagrangian associated with~(\ref{eq:mdp-relaxed}), 
\begin{align}\label{eq:lagrangian}
 &L(\pi,W)\nonumber\\
 &:=\limsup_{T\rightarrow\infty}\frac{1}{T}\mathbb{E}_\pi\!\!\int_{0}^T\!\!\Bigg\{\!\sum_{m=1}^M S_{m,t}\!-\!W\Bigg(B\!-\! \sum_{m=1}^{M}A_{ m,t}\Bigg)\!\Bigg\},
\end{align}
where $W$ is the Lagrangian multiplier, and $\pi$ is a wireless edge caching policy. The corresponding dual function is defined as
\begin{align}
D(W):=\min_\pi L(\pi,W).\label{def:dual}
\end{align}
The dual problem corresponding to $W$ is to optimize the Lagrangian $L(\pi,W)$ over the choice of $\pi$.~For a fixed value of~$W$, the dual problem~\eqref{def:dual} corresponding to the relaxed problem~\eqref{eq:mdp-relaxed} decouples the original problem~\eqref{eq:cmdp-opt1} into $M$ ``\textit{per-content MDPs},'' each of them involving only a single content.  Specifically, the per-content MDP corresponding to the $m$-th content is given as follows,
\begin{align}\label{eq:decoupledduallag}
    \min_{\pi_m}{\bar{C}_{m}}:=\limsup_{T\rightarrow\infty}\frac{1}{T}\mathbb{E}_{\pi_m}\left[\int_{0}^T \bar{C}(S_{m,t}, A_{m,t})dt\right],
\end{align}
where $\bar{C}(S_{m,t}, A_{m,t}):= S_{m,t}-W(1 - A_{m,t})$ is the instantaneous cost incurred by $m$-th content, and $\pi_m$ is a policy that makes decisions (only) for the $m$-th content. It then follows that in order to evaluate the dual function~\eqref{def:dual} at $W$, it suffices to solve all $M$ independent per-content MDPs~(\ref{eq:decoupledduallag}) \cite{puterman1994markov}.
The relaxed problem~(\ref{eq:mdp-relaxed}) can be solved by solving each of these $M$ per-content MDPs, and then combining their solutions.

Unfortunately, this solution does not always provide a feasible wireless edge caching policy for the original problem~(\ref{eq:cmdp-opt1}), which requires that the cache capacity constraint~(\ref{eq:capacity-constraint}) must be met at all times, rather than just in the average sense as in the constraint~(\ref{eq:mdp-relaxed}). Whittle index policy, which we discuss next, combines these solutions corresponding to per-content MDPs in such a way that the resulting policy is also feasible for the original problem~(\ref{eq:cmdp-opt1}), i.e, it satisfies the hard constraint.

\section{Whittle Index Policy}\label{sec:index}
 
We now describe the Whittle index policy that will be utilized for making decisions for wireless edge caching. To the best of our knowledge, Whittle index policy has not been used previously to solve this problem. More specifically, the wireless edge caching problem (\ref{eq:cmdp-opt1}) can be posed as a RMAB problem in which each content $m\in\mathcal{M}$ can be viewed as an arm $m$, and playing arm $m$ would correspond to cache content $m$.  At each time $t$, the queue length $S_{m, t}$ of the corresponding request queue is the state of arm $m$, and $A_{m,t}$ is the action taken for content $m$. $A_{m,t}=1$ denotes that content $m$ is cached at time $t$, and $A_{m,t}=0$ otherwise. It is well-known that Whittle index policy is a computationally tractable solution to the RMAB problem since its computational complexity scales linearly with the number of arms $M$.  For ease of readability, we relegate all proofs in this section to Appendix~\ref{Sec_Appendix_Prop}.

\subsection{Indexability and Whittle Index}
Whittle index policy is defined for a RMAB only when the underlying problem is ``indexable''~\cite{whittle1988restless}. Thus, we begin by showing that our MDP is \textit{indexable}. Loosely speaking, to show that the problem is indexable, we need to consider the single-arm (per-content) MDP~(\ref{eq:decoupledduallag}) and then need to show that the set of states in which the optimal action is passive (i.e., not to cache) increases as the Lagrangian multiplier $W$ increases. We definite it formally here for completeness.

\begin{definition} (Indexability) 
Consider the per-content MDP~(\ref{eq:decoupledduallag}) for the $m$-th content. Let $D_m(W)$ be the set of states in which the optimal action for the per-content MDP~(\ref{eq:decoupledduallag}) is to choose the passive action, i.e., $A_m=0$. Then the $m$-th MDP is said to be indexable if $D_m(W)$ increases with $W$, i.e., if $W>W^\prime$, then $D_m(W)\supseteq D_m(W^{\prime})$. The original MDP~\eqref{eq:cmdp-opt1} is indexable if all of the $M$ per-content MDPs~(\ref{eq:decoupledduallag}) are indexable.
\end{definition} 
In case that a MDP is indexable, the Whittle index for each content/arm is defined as follows. 

\begin{definition} (Whittle Index) 
If the per-content MDP~(\ref{eq:decoupledduallag}) for the $m$-th content is indexable, the Whittle index in state $S$ is the smallest value of the Lagrangian multiplier $W$ such that the optimal policy is indifferent towards actions $A_m=0$ and $A_m=1$ when the Lagrange multiplier is set equal to this value.  We denote this Whittle index by $W_m(S)$, satisfying $W_m(S):=\inf_{W\geq 0}\{S\in D_m(W)\}$.
\end{definition}

\subsection{{The Per-content MDP~(\ref{eq:decoupledduallag})} is Indexable}\label{sec:indexable}
Our proof of indexability relies on the ``threshold'' property of the optimal policy for each per-content MDP~(\ref{eq:decoupledduallag}), i.e., for each $\forall m\in\mathcal{M}$, it is optimal to cache this content only when the number of outstanding requests for it is above a certain threshold; this threshold might depend upon $m$. To show this property, we analyze the per-content MDPs~(\ref{eq:decoupledduallag}).~We begin by analyzing this MDP for a fixed $m$, and thus drop the subscript $m$ in the rest of this subsection for ease of exposition.

\subsubsection{Threshold property of an optimal policy}
We start by analyzing an associated discounted cost MDP, rather than the average latency problem.  After analyzing the discounted MDP, we extend our results to the case of average latency problem~(\ref{eq:decoupledduallag}).~The discounted latency problem corresponding to \eqref{eq:decoupledduallag} is given as follows,
\begin{align}\label{eq:discount_value}
    \min_\pi \mathbb{E}_{\pi}\left[\lim_{T\rightarrow \infty}\int_{0}^T\alpha^{t-1}\bar{C}(S_t, A_t)dt|S_0=s \right],
\end{align}
where $\alpha\in(0,1)$ is a discount factor.  It is well-known that there exists an optimal stationary deterministic policy for this discounted latency problem \cite{puterman1994markov}, and hence we will restrict ourselves to the class of stationary deterministic policies while solving this problem. We apply the value iteration method to find the optimal policy.   

Let $U$ denote the Banach space of bounded real-value functions on $\mathbb{N}$ with supremum norm. Define the operator $\mathcal{T}:U\rightarrow U$ as follows,
\begin{align}\label{eq:operator}
    (\mathcal{T}u)(s):=\min_{a\in\{0,1\}} \bar{C}(s,a)+\alpha\mathbb{E} [u(s^\prime)],
\end{align}
where $u(\cdot)\in U$ and the expectation is taken with respect to the distribution of state $s^\prime$ which results when action $a$ is applied in state $s$. Let $J^\alpha(s)$ denote the optimal expected total discounted cost incurred by the system when it starts in state $s$. Then we have that $J^\alpha(s)=\mathcal{T}J^\alpha(s)$, i.e., $J^\alpha(s)$ is a solution of the Bellman equation satisfying 
\begin{align}\label{eq:bellman1}
    J^\alpha(s)=\min_{a\in\{0,1\}} \bar{C}(s,a)+\alpha\mathbb{E}[J^\alpha(s^\prime)].
\end{align}
As is described in~(\ref{eq:queuetrans}), $s^\prime$ can only assume values from the set $\{s-1,s+1\}$. Let $P_s:=\frac{\lambda}{\lambda+\nu sa}$, so that~(\ref{eq:bellman1}) can be written compactly as $J^\alpha(s)=$ 
\begin{align}\label{eq:bellman2}
 \min_{a\in\{0,1\}}\!\!\bar{C}(s,a)\!+\!\alpha\Big(P_sJ^\alpha(s\!+\!1)\!+\!(1-P_s)J^\alpha(s\!-\!1)\Big).
\end{align}
Define the state-action value function $\forall s\in\cS, a\in \{0,1\}$ as: 
\begin{align}
  Q^\alpha(s,a)&:= \bar{C}(s,a)\nonumber\allowdisplaybreaks\\
  &+\alpha\Big( P_sJ^\alpha(s+1)+(1-P_s)J^\alpha(s-1)\Big).  
\end{align}
Therefore, we have $J^\alpha(s)=\min_{a\in\{0,1\}}Q^\alpha(s,a).$

We need the following assumption on the underlying MDPs in order to ensure that the Whittle indices $W_m(s)$ are finite.
\begin{assumption}\label{assumption1}
There exists a finite $W>0$ such that there is at least one state $s$ in which the optimal action is to activate the arm (cache the content)\footnote{This assumption is valid and such a state always exists.  Otherwise, the optimal action is $a=0$ for any state $s$. From~(\ref{eq:bellman2}), we have $J^\alpha(s)=s-W+\alpha J^\alpha(s+1), \forall s$.  It is straightforward to show that such a recursion results in a non-decreasing $J^\alpha(s)$ in $s$. Hence, there exists a lowest value of state $s$ satisfying $Q^\alpha(s,1)\leq Q^\alpha(s,0)$ for any finite $W.$}, i.e. $Q^\alpha(s,1)\leq Q^\alpha(s,0)$. 
\end{assumption}

We now show that for each value of $W$, the optimal policy for the per-content MDP~\eqref{eq:discount_value} is of threshold-type.

\begin{proposition}\label{prop:threshold-policy}
Consider the discounted latency MDP~\eqref{eq:discount_value} with a fixed $W\ge 0$. There exists an optimal policy for~\eqref{eq:discount_value} that is of threshold-type with the threshold depending on $W$.
\end{proposition}

\begin{remark}
Existing works \cite{singh2021user, hsu2019scheduling} among others have also used the threshold structure of an optimal policy in order to show that the underlying MDP is indexable. The key is to show that if the optimal action for state $s$ is to keep the arm active ($a=1$), i.e., $Q^\alpha(s,1)\leq Q^\alpha(s,0)$,~then the optimal action for state $s+1$ is also to keep it active, i.e., $Q^\alpha(s+1,1)\leq Q^\alpha(s+1,0)$.  The threshold structure in turn is shown by considering the corresponding discounted MDP, and proving for this discounted problem that its value function $Q^\alpha(\cdot)$ of the underlying MDP is convex \cite{singh2021user}, or monotone  \cite{hsu2019scheduling}. Works such as \cite{singh2021user, hsu2019scheduling} showed that these properties hold, but then the associated MDPs in these works have transition rates that are not a function of state. In contrast, the transition rates in our MDPs are a function of state, and hence we cannot use existing results directly.
\end{remark}

The following proposition extends our results in Proposition~\ref{prop:threshold-policy} for the discounted latency problem in \eqref{eq:discount_value} to the original average latency problem in \eqref{eq:decoupledduallag}.

\begin{proposition}\label{Prop:2}
There exists an optimal stationary policy of the threshold-type for the average latency problem in \eqref{eq:decoupledduallag}.
\end{proposition}

\subsubsection{Indexability of the per-content MDP~\eqref{eq:decoupledduallag}}
We now show that the per-content MDP~\eqref{eq:decoupledduallag} is indexable. 
\begin{proposition}\label{prop:indexable}
The per-content MDP~\eqref{eq:decoupledduallag} is indexable. 
\end{proposition}

\begin{proposition}\label{prop:whittle-index-closed}
{Let $\{\phi_R(s)\}_{s=0}^{S_{max}}$ be the stationary distribution of the Markov process which results when the threshold policy with threshold value $R$ is applied.} If the function $\frac{\sum_s s\phi_R(s)- \sum_s s\phi_{R-1}(s) }{\sum_{s=0}^{R} {\phi_{R}(s)}-\sum_{s=0}^{R-1} \phi_{R-1}(s) }$ is non-decreasing in $R$, then the Whittle indices of the per-content MDP~\eqref{eq:decoupledduallag}  are given as follows,
\begin{align}\label{eq:whittle-index1}
W(R):=\frac{\sum_s s\phi_R(s)- \sum_s s\phi_{R-1}(s) }{\sum_{s=0}^{R} {\phi_{R}(s)}-\sum_{s=0}^{R-1} \phi_{R-1}(s) }.
\end{align}
\end{proposition}
From~(\ref{eq:whittle-index1}), it is clear that the stationary distribution of the threshold policy is required to compute the Whittle indices. 

\begin{proposition}\label{prop:stationary-distributions}
The stationary distribution $\{\phi_{R}(s)\}_{s=0}^{S_{max}}$ of the threshold policy with threshold value $R$ satisfies
\begin{align}\label{eq:steadystatedist}\nonumber
    \phi_R (s) &= 0,~s =0,1,\ldots,R-1,\\
     \qquad \phi_R(R)&=\frac{\nu (R+1)}{\lambda+\nu (R+1)}\cdot\phi_R(R+1), \nonumber\\
\phi_R(R+1)&=1/\Bigg(1+\frac{\nu (R+1)}{\lambda+\nu (R+1)}\nonumber\allowdisplaybreaks\\
&+\sum\limits_{l=2}^{S_{max}-R}\prod_{j=2}^l \frac{\lambda}{\lambda+\nu (R+j-1)}\frac{\lambda+\nu (R+j)}{\nu (R+j)}\Bigg),\nonumber\displaybreak[1]\\ 
\phi_R(R+l)&=\phi_{R}(R+1)\prod_{j=2}^l \frac{\lambda}{\lambda+\nu (R+j-1)}\nonumber\allowdisplaybreaks\\
&\cdot\frac{\lambda+\nu (R+j)}{\nu (R+j)}, ~l=2,\cdots, S_{max}-R.
\end{align}
\end{proposition}

\subsection{Whittle Index Policy}
We now describe how the solutions to the relaxed problem~(\ref{eq:mdp-relaxed}) are used to obtain a policy for the original problem~(\ref{eq:cmdp-opt1}). Whittle index policy assigns an index $W_m(S_{m,t})$ to the queues of each content $m\in\mathcal{M}$. This index $W_m(S_{m,t})$ depends upon current state $S_{m,t}$ and current time. The Whittle index policy then activates (caches) $B$ arms (contents) with the highest value of the indices $W_m(S_{m,t})$. Although this policy need not to be optimal for the original problem~(\ref{eq:cmdp-opt1}), it has been shown to be asymptotically optimal~\cite{weber1990index,verloop2016asymptotically} as the number of contents $M$ and the cache size $B$ are scaled up, while keeping their ratio as a constant.

\section{Whittle Index based Q-learning with LFA}\label{sec:learning}

In order to implement the Whittle index policy that was discussed in Section~\ref{sec:index}, one needs to know the controlled transition probabilities of each of the $M$ per-content MDPs. Since this information is often not available, and moreover these parameters are time-varying, we now develop reinforcement learning algorithms that learn the Whittle index policy for wireless edge caching. Specifically, we design a model-free reinforcement learning augmented algorithm with linear function approximation~(LFA), which we call \qwhittleLFA, which leverages the threshold structure of the optimal policy developed in Section~\ref{sec:index} while learning  Q-functions for different state-action pairs.  Similar to Section~\ref{sec:index}, we focus on learning the Whittle index for each per-content MDP, and hence drop the subscript $m$  for ease of presentation.

\subsection{Preliminaries} 
We first review some preliminaries for Q-learning for Whittle index policy, which was first proposed in \cite{fu2019towards} for the discounted cost setup and further generalized in \cite{avrachenkov2020whittle} for average cost setup.  

Consider the dynamic programming equations associated with the per-content MDP in \eqref{eq:decoupledduallag},
\begin{align}\label{eq:dynamic_programming}
 V(s)\!+\!\tilde{\beta}^*&=\min_{a\in\{0,1\}}\Bigg\{a\Big(s+\sum_{s^\prime}p(s^\prime |s,1)V(s^\prime)\Big) \nonumber\allowdisplaybreaks\\
 &+(1-a)\Big(s-W+\sum_{s^\prime}p(s^\prime |s,0)V(s^\prime)\Big)\Bigg\}, 
\end{align}
where $\tilde{\beta}^* \in\mathbb{R}$ is the optimal value of the long-term average cost of the MDP with the Lagrange multiplier set equal to $W$, and $V(\cdot)$ is the relative value function. The corresponding Q-function is given as follows \cite{bertsekas1995dynamic},
\begin{align}\label{eq:Q_value}
    Q(s,a)\!+\!\tilde{\beta}^*\!=s\!-\!(1-a)W(s)\!+\! \sum\nolimits_{s^\prime}p(s^\prime |s,a)V(s^\prime),
\end{align}
where value function $V(\cdot)$ satisfies $V(s)=\min_{a\in\{0,1\}} Q(s,a)$. 
We now discuss a relation satisfied by the Whittle indices $\{W(s)\}_{s\in\cS}$, that was derived in~\cite{fu2019towards}. When the Lagrange multiplier $W$ is set equal to the Whittle index $W(s)$, actions $0$ and $1$ are equally favorable in state $s$, i.e., $Q(s, 0)= Q(s,1)$. Substituting for $Q(s,a)$ from~(\ref{eq:Q_value}) into the relation $Q(s, 0)= Q(s,1)$, we obtain the following relation for $W(s)$,
\begin{align}\label{eq:Whittle index_formulation}
  W(s) = \sum\nolimits_{s^\prime}p(s^\prime|s,0)V(s^\prime)- \sum\nolimits_{s^\prime}p(s^\prime|s,1)V(s^\prime). 
\end{align}

The work~\cite{avrachenkov2020whittle} proposed a tabular Whittle index-based Q-learning algorithm, which we call \qw for ease of exposition. The key aspect of \qw is that the updates of Q-function values and Whittle indices form a two-timescale stochastic approximation (2TSA), where the Q-function values are updated at a faster timescale for a given $W(s)$, and the Whittle indices are updated at a slower timescale.  More precisely, the Q-function values are updated as follows, 
\begin{align}\label{eq: Convention_Q_update}
     &Q_{n+1}(s,a)\!=\!Q_{n}(s,a)\!+\!\gamma_{n} \mathds{1}_{\{S_{n}=s,A_{n}=a\}}\Big(s-(1-a)W(s)\nonumber\\
     &\!+\!\max_{a}Q_{n}(S_{n+1},a)\!-\!I(Q_n)\!-\!Q_{n}(s,a)\!\Big),~n\!=\!0,1,\ldots,
\end{align}
where the subscript $n$ denotes the decision epoch for the  per-content MDP in \eqref{eq:decoupledduallag}, and $I(\cdot)$ is a reference function~\cite{puterman1994markov,abounadi2001learning}. Recall that decision epoch represents the moment when state of the per-content MDP changes. Note that reference functions are used only when performing relative Q-learning iterations for the average cost setup, and not used while optimizing cumulative discounted rewards.  $\{\gamma_{n}\}$ is a step-size sequence satisfying $\sum_n\gamma_{n}=\infty$ and $\sum_n\gamma_{n}^2<\infty$. 
Accordingly, the Whittle indices are updated as follows,
\begin{align}\label{eq:convention_lambda_update}
    W_{n+1}(s)=W_{n}(s)+\eta_{n}(Q_{n}(s,0)-Q_{n}(s,1)),
\end{align}
with the step-size sequence $\{\eta_{n}\}$ satisfying $\sum_n\eta_{n}=\infty$, $\sum_n\eta_{n}^2<\infty$ and $\eta_{n}=o(\gamma_{n})$. The coupled iterates~(\ref{eq: Convention_Q_update}) and~(\ref{eq:convention_lambda_update}) form a 2TSA, and a rigorous asymptotic convergence guarantee is provided in \cite{avrachenkov2020whittle}.

\subsection{\qwhittle}
While \cite{avrachenkov2020whittle} proposed a Q-learning based algorithm for learning Whittle indices, \qw requires a reference function $I(Q_n)$ in order to approximate the unknown parameter $\tilde{\beta}^*$. It is not clear how one should choose the reference function, since there is no unique choice and this function might be problem dependent. To circumvent this problem, a widely-adopted approach is to instead learn an optimal policy for the corresponding discounted-cost MDP, that differs from the average cost MDP only in that the future rewards are discounted. It follows from classical results on MDPs~\cite{blackwell1962discrete} that there exists a stationary deterministic policy that is optimal for all values of discount factor $\alpha$ that are sufficiently close to $1$.  Moreover this policy is also optimal for the average-cost MDP. This policy is known as the Blackwell optimal policy. Such a technique has been applied to the study of average-cost MDP in \cite{fu2019towards,wei2020model}. We will adopt a similar approach, and hence now focus on the discounted Q-learning.

We now use the structural result regarding an optimal policy for the per-content MDP \eqref{eq:decoupledduallag} in order to reduce the exploration overhead associated with the updates of Q-functions.  Specifically, by specializing the Q-learning iterations for a threshold policy with threshold $R$, one needs to update the Q-function values $Q(s,0)$ only for states $s=1,2,\ldots,R-1$ (and not for states $s\ge R$), while all other state-action values are left unchanged since the optimal action for all $s<R$ is deterministic, i.e., $a=0$. Similarly, for action $a=1$, we only need to update the Q-function values $Q(s,1)$ for $s>R$. When the arm is in state $R$, it randomizes between actions $0$ and $1$.  To keep the discussion simple, we assume that these two actions are chosen with equal probability when state is $R$. This key observation drastically reduces the complexity of Q-learning when it is applied to learn Whittle indices, as compared with the existing \qw  \cite{avrachenkov2020whittle}. Towards this end, we call this improved version of \qw algorithm, one which leverages the threshold structure of the optimal
policy, as \qwhittle.

Specifically, we consider the problem of learning the Whittle index for state $s=R$, and develop a recursive update scheme for learning it.  Let $Q^R_n(S_n, A_n)$ be the Q-function value during iteration $n$ with dependence on $R$.
The Q-function updates of \qwhittle are given as follows:

\noindent \textit{Case 1:} When $S_n > R$, we have 
\begin{align}\label{eq:Q_update1}
 & Q^R_{n+1}(S_n,1)\leftarrow
  (1-\gamma_n)Q^R_n(S_n,1)+\gamma_n S_n \nonumber\allowdisplaybreaks\\
  &+\gamma_n \Bigg(\underset{\text{Term}_{11}}{\underbrace{\alpha\mathds{1}_{(S_{n+1}> R)}Q^R_{n}(S_{n+1}, 1)}}+\underset{\text{Term}_{12}}{\underbrace{\alpha\mathds{1}_{(S_{n+1}< R)}Q^R_{n}(S_{n+1}, 0)}}\nonumber\allowdisplaybreaks\\
  &\qquad\qquad+\underset{\text{Term}_{13}}{\underbrace{\alpha\mathds{1}_{(S_{n+1}= R)}\min_{a}Q^R_{n}(S_{n+1}, a)}}\Bigg),
\end{align}
where the step-size sequence $\{\gamma_n\}$ satisfies $\sum_n \gamma_n=\infty$ and  $\sum_n \gamma_n^2<\infty$.  {$\text{Term}_{11}$ follows from the above insight that only Q-function values for states greater than $R$, i.e. $Q^R_n(S_n, 1), S_n>R$ need to be updated.~This differs significantly from \qw \cite{avrachenkov2020whittle},  where both $Q^R_n(\cdot, 1)$ and $Q^R_n(\cdot, 0)$ need to be updated when $S_n>R$. } This is due to the fact that our \qwhittle leverages the threshold-type optimal policy while performing Q-function updates, which either does not exist or is not leveraged in \cite{avrachenkov2020whittle,fu2019towards}.  Similar insights lead to the updates of $\text{Term}_{12}$ and $\text{Term}_{13}$.

\noindent \textit{Case 2:} When $S_n < R$, we have
\begin{align}\label{eq:Q_update2}
   & Q^R_{n+1}(S_n,0)\leftarrow (1-\gamma_n)Q^R_n(S_n,0)+\gamma_n (S_n-W_n(R))\nonumber\allowdisplaybreaks\\
    &+\gamma_n\Bigg(\underset{\text{Term}_{21}}{\underbrace{\alpha\mathds{1}_{(S_{n+1}>R)}{Q}^R_{n}(S_{n+1},1)}}+\underset{\text{Term}_{22}}{\underbrace{\alpha\mathds{1}_{(S_{n+1}<R)}Q^R_{n}(S_{n+1},  0)}}\nonumber\allowdisplaybreaks\\
    &\qquad\qquad+\underset{\text{Term}_{23}}{\underbrace{\alpha\mathds{1}_{(S_{n+1}=R)}\min_{a}Q^R_n(S_{n+1},  a)}}
 \Bigg),
\end{align}
where the updates of $\text{Term}_{21}$, $\text{Term}_{22}$ and $\text{Term}_{23}$ leverage similar insights as those in Case 1.

\noindent \textit{Case 3:} When $S_n = R$, $Q^R_{n}(S_{n},A_n)$ gets updated according to either \eqref{eq:Q_update1} or \eqref{eq:Q_update2} with equal probability.  

In summary, the Q-function updates of \qwhittle are given as 
\begin{align}\label{eq:Q_update}
  Q^R_{n+1}(s,a)=
\begin{cases}
 \eqref{eq:Q_update1}~\text{or}~\eqref{eq:Q_update2},~\text{if $(s,a)=(S_n,A_n),$}\\
  Q_n^R(s,a),~\text{otherwise.} 
\end{cases}
\end{align}
With the above Q-function updates, the parameter $W$ under the threshold policy with threshold $R$ is updated as follows,
\begin{align}\label{eq:W_update}
    W_{n+1}(R)=W_n(R)+\eta_n\Big(Q^R_n(R,0)-Q^R_n(R,1)\Big),
\end{align}
where the step-size sequence $\{\eta_n\}$ satisfies $\sum_n \eta_n=\infty$, $\sum_n \eta_n^2<\infty$ and $\eta_{n}=o(\gamma_{n})$. 

\qwhittle is summarized in Algorithm \ref{Algorithm2}. Q-function and $W$ updates discussed above remain the same for all $M$ contents (lines 4-8). Since the wireless edge can cache at most $B$ contents, an easy implementation is to find the possible activation set $\mathcal{C}:=\{m\in \mathcal{M}|S_m(t)\geq R\}$ for threshold $R$  at time/epoch $t$ and activate $\min(B,|\mathcal{C}|)$ arms with highest Whittle indices $W_{m,t}(S_{m,t})$. Note that $t$ is the moment when the state of any of the $M$ per-content MDPs changes.

\begin{algorithm}[t]
\caption{\qwhittle for Per-Content MDP}
\label{Algorithm2}
\begin{algorithmic}[1]
\STATE Initialize: ${Q}^{s^\prime}_{0}(s,a)=0,~ W_{0}(s)=0$, $\forall s, s^\prime\in\mathcal{S}$.
\FOR {$R\in\mathcal{S}$}
\STATE Set the threshold policy as $\pi=R$.
\FOR{$n=1,2,\ldots, T$}
\STATE Update ${Q}^R_{n}(s_{n},a_{n})$ according to~\eqref{eq:Q_update}.
\STATE Update $W_{n}(R)$ according to~(\ref{eq:W_update}).
\ENDFOR
\STATE $W_{0}(R+1)=W_{T}(R)$, ${Q}_{0}^{R+1}(s,a)={Q}_{T}^{R}(s,a)$.
\ENDFOR
\STATE Return: $W(s), \forall s\in\mathcal{S}$.
\end{algorithmic}
\end{algorithm}

\begin{remark}\label{rem:q-learning-comparison}
Some definitions (e.g., $W(s)$) in this paper are similar to those in \cite{avrachenkov2020whittle,fu2019towards}, which studied \qw through a two-timescale update. However, our \qwhittle differs from those in \cite{avrachenkov2020whittle,fu2019towards} from two perspectives. First, \cite{avrachenkov2020whittle,fu2019towards} adopted the conventional $\epsilon$-greedy rule for Q-function value updates.  In contrast, we leverage the property of optimal threshold-type policy into Q-function value updates as in \eqref{eq:Q_update1} and \eqref{eq:Q_update2}.  Such a threshold-type Q-function value update dramatically reduces the computational complexity  since each state only has a fixed action to explore.  Second, the threshold policy further enables us to update Whittle indices in an incremental manner, i.e., the converged Whittle index in state $s$ can be taken as the initial value for the subsequent state $s+1$ (line 8 in Algorithm \ref{Algorithm2}), instead of being randomly initiated as in \cite{fu2019towards,avrachenkov2020whittle}.  This further speeds up the learning process.  
In addition, \cite{fu2019towards} lacked convergence guarantee. 
Recently, another line of work \cite{biswas2021learn} leveraged Q-learning to approximate Whittle indices through a single-timescale SA, where Q-function and Whittle indices were learned independently.  \cite{biswas2021learn} considered  the finite-horizon MDP and cannot be directly applied to infinite-horizon discounted or average cost MDPs considered in this paper.    
\end{remark}

\subsection{\qwhittle with Linear Function Approximation}\label{sec:lfa}
When the number of state-action pairs is very large, which is often the case for wireless edge caching, \qwhittle can be intractable due to the curse of dimensionality.  A closer look at~(\ref{eq:W_update}) further reveals that the Whittle index is updated only when state $s$ is visited. This can significantly slow down the convergence process of the corresponding 2TSA when the state space is large.  To overcome this difficulty, we further study \qwhittle with linear function approximation (LFA) by using low-dimensional linear approximation of $Q$ on a linear subspace with {dimension $d\ll |\mathcal{S}||\mathcal{A}|$.}  We call this algorithm as \qwhittleLFA.

Specifically, given a set of basis functions $\phi_\ell: \cS\times \cA\mapsto \mathbb{R},~\ell=1,\cdots, d$, the approximation of the Q-function $\tilde{Q}_{\theta}(s,a)$ parameterized by a unknown weight vector ${\theta}\in\mathbb{R}^d$, is given by $\tilde{Q}_{\theta}(s,a)={\phi}(s,a)^{\intercal} {\theta}, ~\forall s\in\mathcal{S}, a\in\mathcal{A},$ where ${\phi}(s,a)=(\phi_1(s,a), \cdots,\phi_d(s,a))^\intercal$. The feature vectors are assumed to be linearly independent and are normalized so that $\|{\phi}(s,a)\|\leq 1, \forall s\in\mathcal{S}, a\in\mathcal{A}$.

Similar to \qwhittle, we consider the problem of learning the Whittle index for state $s=R$, which can be equivalently formulated as the problem of learning the coefficient $\theta$.  Let $\theta^R_n$ be its value during iteration $n$, which depends on the value of $R$.  Leveraging the same ideas for Q-function updates in~(\ref{eq:Q_update}), \qwhittleLFA iteratively updates $\theta_n^R$ as follows:


\noindent \textit{Case 1:} When $S_n> R$, we have  
\begin{align}\label{eq:theta_update1}
  &\theta^R_{n+1}\!\leftarrow\!
 \theta^R_n+\gamma_n \phi(S_n,1)\Bigg[S_n +{{\alpha\mathds{1}_{(S_{n+1}> R)}\phi(S_{n+1}, 1)^\intercal \theta^R_n}}\nonumber\allowdisplaybreaks\\
 &+{{\alpha\mathds{1}_{(S_{n+1}< R)}\phi_{n}(S_{n+1}, 0)^\intercal \theta^R_n}}\nonumber\displaybreak[1]\\
  &\!+\!{{\alpha\mathds{1}_{(S_{n+1}= R)}\min_{a}\phi(S_{n+1}, a)^\intercal \theta^R_n}}\!\!-\!\phi(S_{n},1)^\intercal \theta^R_n\Bigg].
\end{align}

\noindent \textit{Case 2:} When $S_n< R$, we have,
\begin{align}\label{eq:theta_update2}
   \hspace{-0.3cm} &\theta^R_{n+1}\!\!\leftarrow\! \theta^R_n+\gamma_n\phi(S_n,0)\Bigg[ (S_n-W)\nonumber\allowdisplaybreaks\\
    &+{{\alpha\mathds{1}_{(S_{n+1}>R)}{\phi}_{n}(S_{n+1},1)^\intercal\theta^R_n}}+{{\alpha\mathds{1}_{(S_{n+1}<R)}\phi_{n}(S_{n+1},  0)^\intercal\theta^R_n}}\nonumber\displaybreak[1]\\
    &+{{\alpha\mathds{1}_{(S_{n+1}=R)}\min_{a}\phi_n(S_{n+1},  a)^\intercal \theta^R_n}}\!\!-\!\phi(S_{n},0)^\intercal \theta^R_n\Bigg].
\end{align}

\noindent \textit{Case 3:} When $S_n= R$, the update occurs either according to  \eqref{eq:theta_update1} or \eqref{eq:theta_update2} with an equal probability. 

The iterations can be summarized as follows, 
\begin{align}\label{eq:theta_update}
  \theta^R_{n+1}=
\begin{cases}
 \eqref{eq:theta_update1}~\text{if $S_n> R$},\\
 \eqref{eq:theta_update2}~\text{if $S_n< R$},\\
 \eqref{eq:theta_update1}~\text{or }\eqref{eq:theta_update2},~\text{if $S_n= R$}.
\end{cases}
\end{align}
We now derive a similar iterative scheme for learning the Whittle indices. To do this, we consider the Whittle index update in~(\ref{eq:W_update}), and replace the Q-function values $Q^R_n(R,0),Q^R_n(R,1)$ by their linear function approximations $\phi(R,0)^{\intercal}\theta^R_n$ and $\phi(R,1)^{\intercal}\theta^R_n$, respectively. This gives us the following iterations,
\begin{align}    \label{eq:lambda_update}
    W_{n+1}(R)=W_{n}(R)+\eta_n({\phi}(R,0)^{\intercal} {\theta}^R_n-{\phi}(R,1)^{\intercal} {\theta}^R_n).
\end{align}
We summarize \qwhittleLFA in Algorithm \ref{Algorithm3}, which is one of our key contributions in this paper.

\begin{algorithm}[t]
\caption{ \qwhittleLFA for Per-Content MDP }
\label{Algorithm3}
\begin{algorithmic}[1]
\STATE Initialize: ${\phi}(s,a),~ {\theta}_{0}, W_{0}(s)=0$ for $\forall s\in\mathcal{S}, a\in\mathcal{A}$.
\FOR{$R\in\mathcal{S}$}
\STATE Set the threshold policy as $\pi=R$.
\FOR{ $n=1,2,\ldots, T$}
\STATE Update ${\theta}^R_{n}$ according to~\eqref{eq:theta_update}.
\STATE Update $W_{n}(R)$ according to~(\ref{eq:lambda_update}).
\ENDFOR
\STATE $W_{0}(R+1)=W_{T}(R)$, ${\theta}_{0}^{R+1}(s,a)={\theta}_{T}^{R}(s,a)$.
\ENDFOR
\STATE Return: $W(s), \forall s\in\mathcal{S}$.
\end{algorithmic}
\end{algorithm}

\section{Finite-Time Performance Analysis}\label{sec:finite-time-analysis}

In this section, we provide a finite-time analysis of our \qwhittleLFA algorithm, which can be viewed through the lens of 2TSA.  Our key technique is motivated by \cite{doan2020nonlinear}, which deals with  a general nonlinear 2TSA.  To achieve this goal, we first need to rewrite our \qwhittleLFA updates in \eqref{eq:theta_update} and \eqref{eq:lambda_update} in the form of a 2TSA. Throughout this section, we will perform the analysis for any threshold policy $\pi=R$, and hence we drop the superscript $R$ for ease of presentation.

\subsection{Two-Timescale Stochastic Approximation}

Given the threshold policy $\pi=R$, the corresponding true Whittle index associated with the threshold state $R$ is $W(R)$.  Denote $\theta^R$ as the optimal $\theta$ obtained by \qwhittleLFA in Algorithm \ref{Algorithm3}. Following the conventional ODE method~\cite{borkar2009stochastic}, we begin by converting \qwhittleLFA in \eqref{eq:theta_update} and \eqref{eq:lambda_update} into a standard 2TSA. In particular, we rewrite the updates \eqref{eq:theta_update} and \eqref{eq:lambda_update} as follows,
\begin{align}
    \theta_{n+1}&=\theta_n+\gamma_n[h(\theta_n, W_n)+\xi_{n+1}],\label{eq:Q_2TSA}\\
    W_{n+1}&=W_n+\eta_n[g(\theta_n,W_n)+\psi_{n+1}],\label{eq:W_2TSA}
\end{align}
where $\{\xi_n\}$ is an appropriate martingale difference sequence with respect to the filtration $\sigma$-field $\mathcal{F}_n=\{\theta_0,W_0,\xi_0, \ldots, \theta_n, W_n, \xi_n\}, n=1,2,\ldots$; $\{\psi_n\}$ is a suitable error sequence;  $h$ and $g$ are appropriate Lipschitz functions defined below that satisfy the conditions needed for our ODE analysis, and the step sizes $\gamma_n,\eta_n$ satisfy Assumption~\ref{assumption:stepsize} below. Note that the $\theta_n$ and $W_n$ iterations are coupled. By using the operator defined in~(\ref{eq:operator}), we rewrite the $\theta$ update in \eqref{eq:theta_update} as follows,
\begin{align}
&\theta_{n+1}=\theta_{n}+\gamma_n\phi(S_n,A_n)\Big[[\mathcal{T}{\theta}_n](S_n,A_n)\nonumber\allowdisplaybreaks\\
    &-\phi(S_n,A_n)^\intercal\theta_{n}+\xi_{n+1}(S_n,A_n)\Big],~\forall (S,A)\in\mathcal{S}\times\mathcal{A},
\end{align}
where, 
\begin{align}
    [\mathcal{T}{\theta}_n](S_n,A_n)&=S_n-(1-A_n)W_n\nonumber\allowdisplaybreaks\\
    &+\alpha\sum_{s^\prime}p(s^\prime|S_n,A_n)\min_{a^\prime}\phi(s^\prime, a^\prime)^\intercal\theta_n,\\
    \xi_{n+1}(S_n,A_n)&\!=\!S_n\!-\!(1\!-\!A_n)W_n\!+\!\alpha\min_{a}\phi(S_{n+1},a)^\intercal\theta_n\nonumber\allowdisplaybreaks\\
    &-[\mathcal{T}{\theta}_n](S_n,A_n).
\end{align}
Hence, we have  $h(\theta_n, W_n)$ in \eqref{eq:Q_2TSA} as 
\begin{align}\label{eq:h-ODE}
    h(\theta_n, W_n):=[\mathcal{T}\theta_n](S_n,A_n)-\phi(S_n,A_n)^\intercal\theta_n,
\end{align}
which is Lipschitz in both $\theta$ and $W$. 
Similarly, we have
\begin{align}\label{eq:g-ODE}
    g(\theta_n, W_n):={\phi}(R,0)^{\intercal} {\theta}_n-{\phi}(R,1)^{\intercal} {\theta}_n,
\end{align}
which is Lipschitz in $\theta.$  W.l.o.g., we assume $\psi_{n}=0$ for all $n$ since the update of $W$ in~(\ref{eq:lambda_update}) is deterministic. After having identified these two functions, i.e., \eqref{eq:h-ODE} and \eqref{eq:g-ODE},  the asymptotic convergence of our 2TSA can be established by using the ODE method 
\cite{borkar2009stochastic,suttle2021reinforcement, avrachenkov2020whittle, doan2019linear,doan2020nonlinear}.  For ease of exposition, we temporally assume fixed step size here, then the 2TSA is reduced to the following differential equations:
\begin{align}\label{eq:2TDE}
    \dot{\theta}(t)=h(\theta(t),W(t)),\quad    \dot{W}(t)&=\frac{\eta}{\gamma}g(\theta(t),W(t)),
\end{align}
where the ratio $\eta/\gamma$ represents the difference in timescale between these two updates.  
Our focus here is on characterizing the finite-time convergence rate of $(\theta_n, W_n)$ to the globally asymptotically optimal equilibrium point $(\theta^R, W(R))$ of~(\ref{eq:2TDE}) for each $R.$ Using an idea of  \cite{doan2020nonlinear}, the key part of our analysis is based upon an appropriate choice of two step sizes $\eta_n, \gamma_n$, and a Lyapunov function.  We first define the following two ``error terms,'' 
\begin{align}\label{eq:residual}
    \tilde{\theta}_n :&=\theta_n-f(W_n),\quad 
    \tilde{W}_n:=W_n-W(R),
\end{align}
which characterizes the coupling between $\theta_n$ and $W_n$. If we are able to show that $\tilde{\theta}_n$ and $\tilde{W}_n$ simultaneously converge to zero, then we would have shown $(\theta_n,W_n)\rightarrow (\theta^R,W(R))$. Thus, to prove the convergence of $(\theta_n,W_n)$ of our 2TSA to its true value $(\theta^R,W(R))$, we instead study the convergence of $(\tilde{\theta}_n,\tilde{W}_n)$ by providing the finite-time analysis for the  mean squared error generated by~(\ref{eq:Q_2TSA})-(\ref{eq:W_2TSA}). In order to simultaneously study the properties of $\tilde{\theta}_n$ and $\tilde{W}_n$, we consider the following Lyapunov function 
\begin{align}\label{eq:lyapunov-function}
    M(\theta_n, W_n)&:=\frac{\eta_n}{\gamma_n}\|\tilde{\theta}_n\|^2+\|\tilde{W}_n\|^2\nonumber\allowdisplaybreaks\\
    &=\frac{\eta_n}{\gamma_n}\|\theta_n-f(W_n)\|^2+\|W_n-W(R)\|^2.
\end{align}
We make the following assumptions while analyzing~(\ref{eq:Q_2TSA})-(\ref{eq:W_2TSA}).  

\begin{assumption}\label{assumption:Lipschitz}
Provided any $W\in\mathbb{R}$, there exists an operator $f$ such that $\theta=f(W)$ is the unique solution to $h(\theta, W)=0,$ where $h$ and $f$ are Lipschitz continuous with positive constants $L_h$ and $L_f$ such that
\begin{align}\nonumber
   & \|f(W)-f(W^\prime)\|\leq L_f\|W-W^\prime\|, \\
    &\|h(\theta,W)\!-\!h(\theta^\prime,W^\prime)\|\leq L_h(\|\theta\!-\!\theta^\prime\|\!+\!\|W\!-\!W^\prime\|).
\end{align}
The operator $g$ in \eqref{eq:W_2TSA}  is Lipschitz continous with constant $L_g$, i.e.,
\begin{align}
    \|g(\theta,W)-g(\theta^\prime,W^\prime)\|\leq L_g(\|\theta-\theta^\prime\|+\|W-W^\prime\|).
\end{align}
\end{assumption}

\begin{remark}
The Lipschitz continuity of the functions $f,g,h$ guarantees the existence of solutions to the ODEs~\eqref{eq:2TDE}.  Note that when $h$ and $g$ are linear functions of $\theta$ and ${W}$, Assumption \ref{assumption:Lipschitz} is automatically satisfied. This assumption is widely used for both linear and nonlinear 2TSA \cite{chen2019performance,gupta2019finite,doan2020nonlinear}. 
\end{remark}

\begin{assumption}\label{assump:stable}
There exist  $\mu_1>0$ and  $\mu_2>0$ such that
\begin{align}
    \tilde{{\theta}}^\intercal h({\theta}, W)&\leq -\mu_1\|\tilde{{\theta}}\|^2,~\forall \theta,\tilde{\theta}\in\mathbb{R}^d,~W\in\mathbb{R},\nonumber\\
    \tilde{W}g({\theta}, W)&\leq -\mu_2\|\tilde{W}\|^2,~\forall \tilde{W},W\in\mathbb{R},\theta\in\mathbb{R}^d.
\end{align}
\end{assumption}
\begin{remark}
Assumption~\ref{assump:stable} guarantees the uniqueness of the solution to the ODEs~\eqref{eq:2TDE}. This assumption can be viewed as a relaxation of the monotone property of the nonlinear mappings \cite{doan2020nonlinear, chen2019performance}, since it is automatically satisfied if $h$ and $g$ are strongly monotone as is assumed in \cite{doan2020nonlinear}. 
\end{remark}

\begin{assumption}\label{assumption:mean-variance}
Random variables $\xi_n$ are independent of each other and across time, with zero mean and bounded variances
\begin{align*}
\mathbb{E}[\xi_n|\mathcal{F}_{n-1}]=0, \qquad
    \mathbb{E}[\|\xi_n\|^2|\mathcal{F}_{n-1}]\le \Lambda,
\end{align*}
where $\Lambda>0$.
\end{assumption}

\begin{assumption}\label{assumption:stepsize}
The step sizes $\gamma_n$ and $\eta_n$ satisfy $\sum_{n=0}^\infty \gamma_n=\sum_{n=0}^\infty \eta_n=\infty$,$\sum_{n=0}^\infty \gamma_n^2<\infty,$ $\sum_{n=0}^\infty \eta_n^2<\infty,$ $\eta_n/\gamma_n$ is non-increasing in $n$ and $\lim_{n\rightarrow\infty}\eta_n/\gamma_n=0.$
\end{assumption}
\begin{remark}
These assumptions are standard in SA literature  \cite{borkar2009stochastic,suttle2021reinforcement, avrachenkov2020whittle, doan2019linear,doan2020nonlinear}. Assumption~\ref{assumption:mean-variance} holds since 
$\xi_{n}(s,a)=s-(1-a)W_n+\alpha\min_{a}\phi(S_{n+1},a)^\intercal\theta_n-[\mathcal{T}{\theta}_n](s,a)$, thus $\mathbb{E}[\xi_n|\mathcal{F}_{n-1}]=0.$ 
\end{remark}

\subsection{Finite-Time Analysis of \qwhittleLFA}

\begin{theorem}\label{thm:convergence}
Consider the iterates $\{{\theta}_n\}$ and $\{{W}_n\}$ generated by~(\ref{eq:theta_update}) and~(\ref{eq:lambda_update}) for learning the Whittle indices, and suppose that Assumptions~\ref{assumption:Lipschitz}-\ref{assumption:stepsize} hold true. Let the step-sizes be chosen as $\gamma_n=\frac{\gamma_0}{(n+1)^{5/9}},\eta_n=\frac{\eta_0}{(n+1)^{10/9}}$. Then we have
\begin{align}\label{eq:bound}
  &\mathbb{E}[M({\theta}_{n+1},W_{n+1})|\mathcal{F}_{n}] \le \frac{\mathbb{E}[M({\theta}_0,W_0)]}{(n+1)^2}
    \nonumber\allowdisplaybreaks\\
    & +\!\frac{C_1(\|\tilde{{\theta}}_0\|^2+\|\tilde{W}_0\|^2)}{(n+1)^{2/3}}\!+\!\frac{\gamma_0\eta_0\Lambda}{(n+1)^{2/3}},~n=1,2,\ldots,
\end{align}
where $C_1=(L_h^2+L_f^2+2L_g^2(L_f+1)^2)\alpha_0\eta_0+2L_g^2(L_f+1)^2\left(L_f^2+{(1+L_h\alpha_0)^2}\right)\frac{\eta_0^3}{\gamma_0^3}$.
\end{theorem}
The first term of the right hand side of (\ref{eq:bound}) corresponds to the bias due to the initialization, which goes to zero at a rate $\mathcal{O}(1\slash n^2)$. The second term corresponds to the accumulated estimation error of the nonlinear 2TSA.  The third term stands for the error introduced due to the fluctuations of the martingale difference noise sequence $\{\xi_n\}$ in \eqref{eq:Q_2TSA}. The second and third terms in the right hand side of (\ref{eq:bound}) decay at a rate $\mathcal{O}(1\slash n^{2/3})$, and hence dominate the overall convergence rate in~(\ref{eq:bound}).  The proof is presented in Appendix~\ref{sec:QWhittle-learning-convergence-app}.

\begin{remark}
Our finite-time analysis of \qwhittleLFA consists of two steps.  First, we rewrite \qwhittleLFA updates into a 2TSA in~(\ref{eq:Q_2TSA})-(\ref{eq:W_2TSA}).  The key is to identify two critical terms $h$ and $g.$ Second, we prove a bound on finite-time convergence rate of \qwhittleLFA by leveraging and generalizing the machinery of nonlinear 2TSA \cite{doan2020nonlinear}. The key is to the choice of two step sizes (as characterized in Theorem~\ref{thm:convergence}) and a Lyapunov function given in~(\ref{eq:lyapunov-function}).  Though the main steps of our proof are motivated by \cite{doan2020nonlinear}, we need to characterize the specific requirements for our settings as aforementioned.  Need to mention that we do not need the assumption that $h$ and $g$ are strongly monotone as in \cite{doan2020nonlinear}, and hence requires a re-derivation of the main results.    
\end{remark}

\section{Numerical Results}\label{sec:evaluation}
In this section, we numerically evaluate the performance of our \qwhittleLFA algorithm using both synthetic and real traces.

\subsection{Baselines and Experiment Setup}

We compare \qwhittleLFA to existing learning based algorithms for wireless edge caching.  In particular, we focus on both Q-learning based Whittle index policy for wireless edge caching (see Remark~\ref{rem:q-learning-comparison}) such as \emph{Q-learning Whittle Index Controller} (QWIC) \cite{fu2019towards}, \qw \cite{avrachenkov2020whittle}, \emph{Whittle Index Q-learning} (WIQL) \cite{biswas2021learn} and \emph{Deep Threshold Optimal Policy Training} (DeepTOP) \cite{nakhleh2022deeptop}; and existing learning based algorithms for wireless edge caching  such as  \emph{Follow-the-Perturbed-Leader} (FTPL) \cite{bhattacharjee2020fundamental}, \emph{Deep Q-Learning} (DQL) \cite{wu2019dynamic} and \emph{Deep Actor-Critic} (DAC) \cite{zhong2020deep}.  We also compare these learning based algorithms to our Whittle index policy (see Section~\ref{sec:index}), which is provably asymptotically optimal when system parameters are known.  For the above algorithms using neural networks, we consider two hidden layers with size (64, 32), with external memory size being 10,000 and batch size being 10. The discount factor is $\alpha = 0.98$. The learning rates are initialized to be $\gamma_0 = 0.1$ and $\eta_0 = 0.01$, and are decayed by 1.1 every $1,000$ time steps.  In LFA, we set $d = 20$.

\begin{figure}[t]
\centering
 \includegraphics[width=0.5\textwidth]{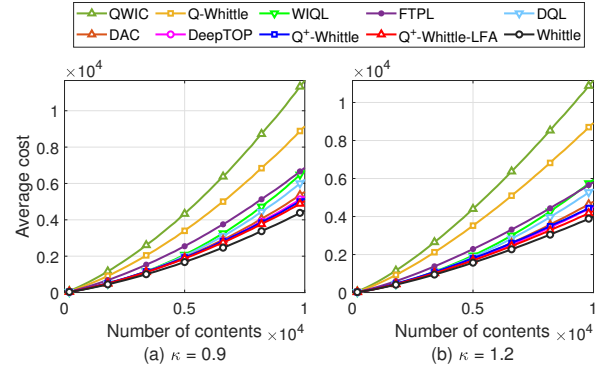}
\caption{Accumulated cost (latency) using synthetic traces. }
\label{fig:learning-cost-synthetic}
\vspace{-0.15in}
\end{figure}

\subsection{Evaluation Using Synthetic Traces}\label{sec:evaluation-sync}

We simulate a system with the number of distinct contents $M$ ranging from $200$ to $10,000$ with a step size of $200$.  In each case, content requests are drawn from a Zipf distribution with  Zipf parameters $\kappa$ of $0.9$ and $1.2$.  As we consider a state-dependent delivery rate in our model~(\ref{eq:queuetrans}),  we set the ``unit rate''  $\nu=18$ with the true delivery rate of $\nu SA,$ and the total number of requests varies across each $M$. The cache size is $B=M/10$.

\begin{figure*}[t]
\centering
\begin{minipage}{.48\textwidth}
\centering
 \includegraphics[width=1\columnwidth]{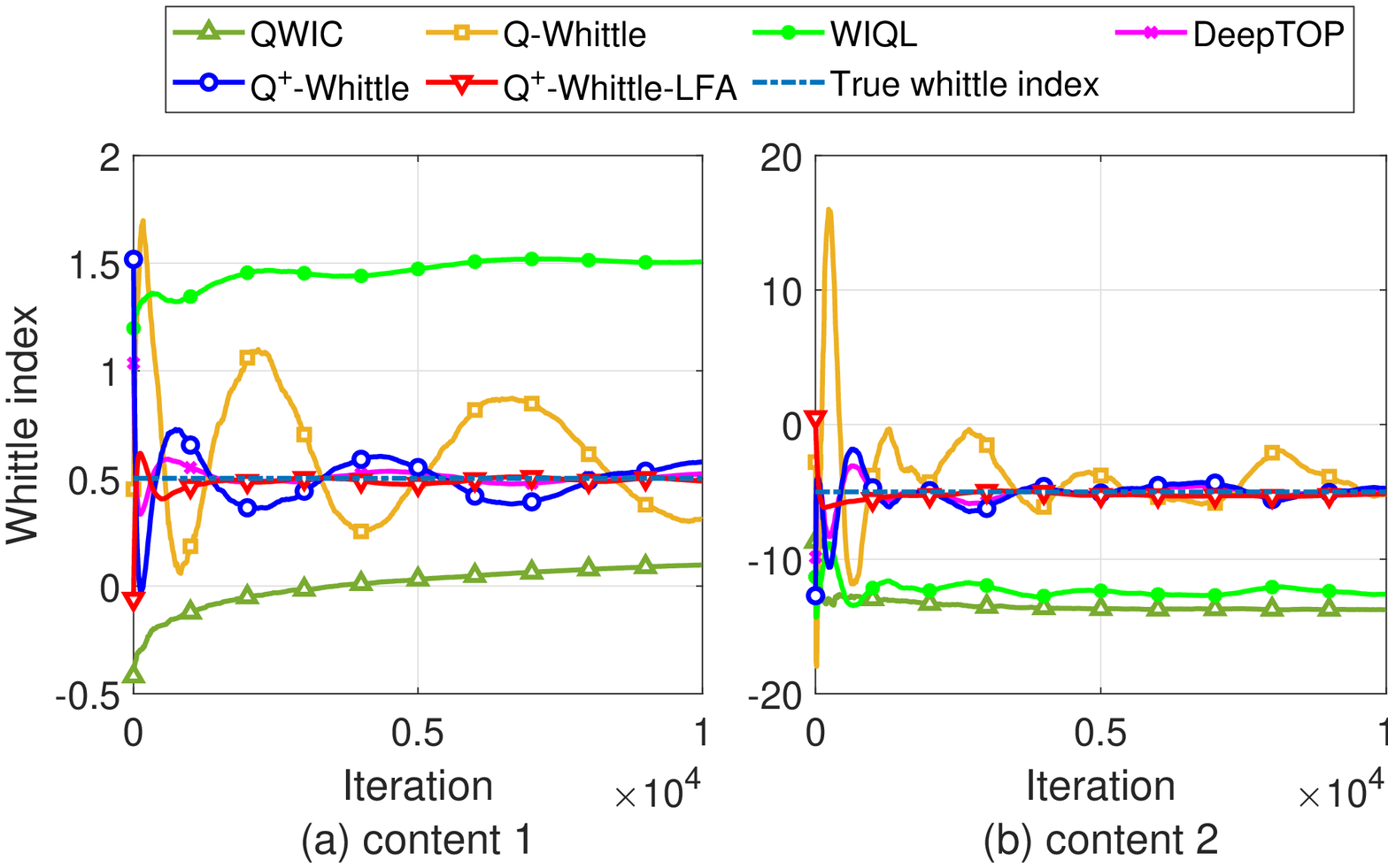}
\caption{Convergence in terms of iterations of Whittle index based Q-learning algorithms for two randomly selected contents.}
\label{fig:convergence}
\end{minipage}\hfill
\begin{minipage}{.48\textwidth}
\centering
 \includegraphics[width=1\columnwidth]{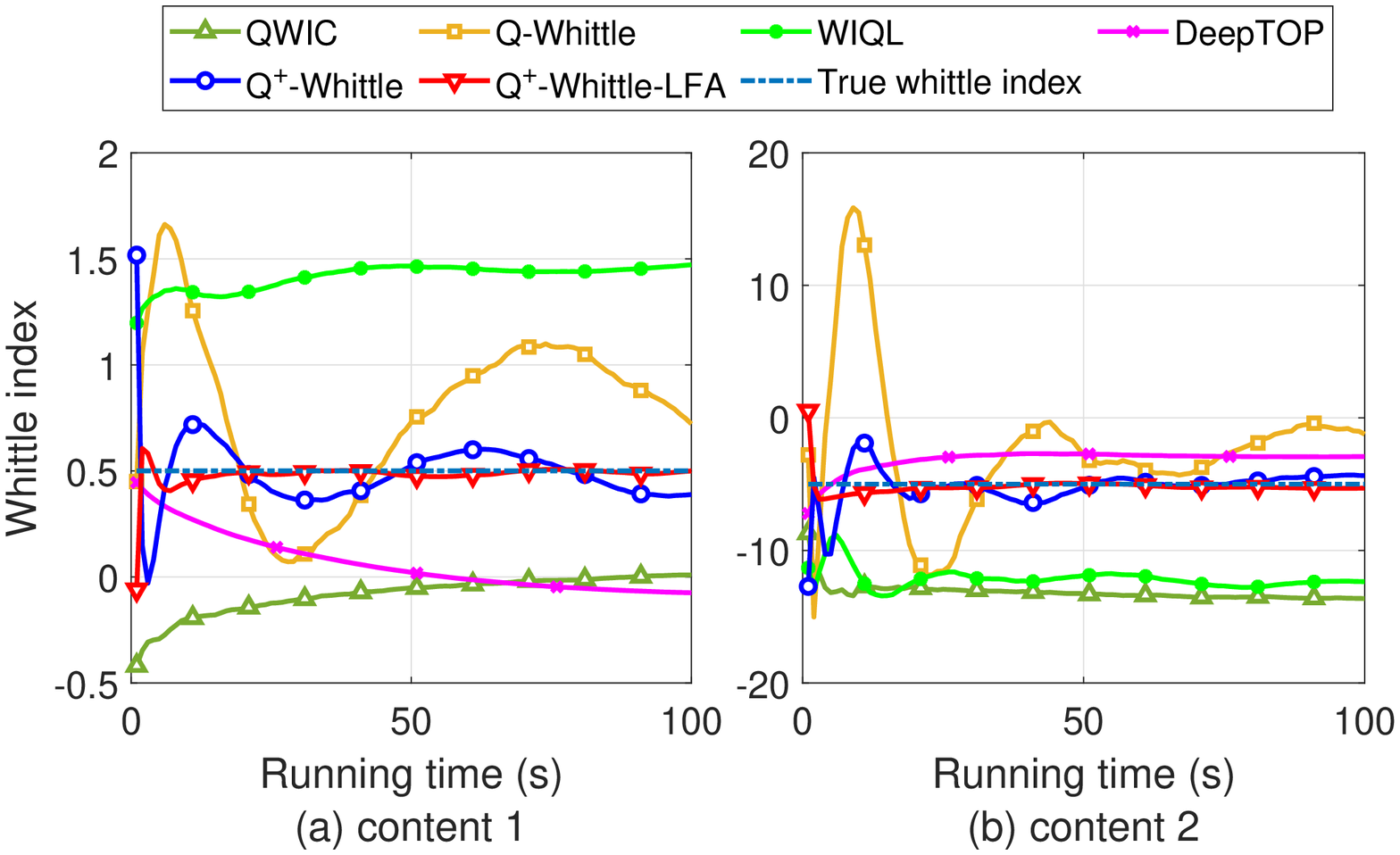}
\caption{Convergence in terms of running time of Whittle index based Q-learning algorithms for two randomly selected contents.}
\label{fig:convergencetime}
\end{minipage}
\vspace{-0.2in}
\end{figure*}

\noindent\textbf{Accumulated Cost (Latency)}. The accumulated costs of above learning based algorithms are presented in  Figure~\ref{fig:learning-cost-synthetic}, where we use the Monte Carlo simulation with $2,000$ independent trails. From Figure~\ref{fig:learning-cost-synthetic}, it is clear that our \qwhittleLFA consistently outperforms its counterparts. In addition, WIQL outperforms QWIC and \qw, which is consistent with the observations made in \cite{biswas2021learn}. Moreover, our \qwhittle and \qwhittleLFA perform close to the Whittle index policy.  This is due to the fact that both leverage the asymptotically optimal Whittle index policy to make decisions for wireless edge caching.  Finally, we remark that \qwhittleLFA is much more computationally efficient compared to \qwhittle, and \qw in \cite{avrachenkov2020whittle}, especially when the state space is large.  This observation is further pronounced when we compare their convergence as illustrated below.

\noindent\textbf{Convergence and Running Time.} We demonstrate the convergence of Whittle index based Q-learning algorithms in terms of the number of iterations in Figure~\ref{fig:convergence}, and in terms of running time in Figure~\ref{fig:convergencetime}.  The running time are obtained via averaging over 2,000 Monte Carlo runs of a single-threaded program on Ryzen 7 5800X3D desktop with 64 GB RAM. In both figures, we randomly draw two contents from the trace with Zipf parameter $0.9$ due to the decoupled nature of our framework (see Section~\ref{sec:learning}). For ease of exposition, we only show results of Whittle indices of two states for these two particular contents.

We observe that the Whittle indices obtained by our \qwhittle and \qwhittleLFA converge to the true Whittle indices, which are obtained under the assumption that system parameters are known.  More importantly, \qwhittleLFA converges much faster than \qwhittle both in terms of iterations and running time, as motivated earlier.  In addition, the \qw in \cite{avrachenkov2020whittle} is provably convergent to the true Whittle index, however, the multi-timescale nature of \qw makes it converges slowly in practice (see discussions in Section~\ref{sec:related}). As shown in Figure~\ref{fig:convergence}, \qw still cannot converge to the true Whittle indices after 10,000 iterations while our \qwhittleLFA converges only after 1,000 iterations.  We note that DeepTOP \cite{nakhleh2022deeptop} also leverages a threshold policy to learn Whittle indices, which converges to the true Whittle indices in a smaller number of iterations as shown in Figure~\ref{fig:convergence} but with a much larger running time as shown in Figure~\ref{fig:convergencetime}.  This is due to its intrinsic nature of training a deep neural network in each iteration for making decisions.  Finally, though QWIC and WIQL may converge, they are not guaranteed to converge to the true Whittle indices as observed in Figure~\ref{fig:convergence}. Similar observations hold for other contents in other traces, and hence are omitted here.

\subsection{Evaluation Using Real Traces}
We further evaluate \qwhittleLFA using two real traces: (i) \emph{Iqiyi} \cite{ma2017understanding}, which contains mobile video behaviors; and (ii) \emph{YouTube} \cite{zink2008watch}, which contains trace data about user requests for specific {YouTube} content collected from a campus network. For the Iqiyi (resp. YouTube) trace, there are more than $67$ (resp. $0.6$) million requests for more than $1.4$ million (resp. $0.3$) unique contents over a period of $335$ (resp. $336$) hours.  We evaluate the accumulated cost over rough $14$ days for each trace with a cache size of $B=4,000$ (resp. $2,000$) for Iqiyi (resp. YouTube).  We choose these values based on the observation of average number of active contents in the traces\footnote{A content is said to be active at time $t$ if $t$ lies between the first and the last requests for the content.}. 
The accumulated costs of the above learning based algorithms are shown in Figure~\ref{fig:learning-cost-real}.  Again, we observe that our \qwhittleLFA  significantly outperforms its counterparts with smaller costs.  Finally, we note that \qwhittleLFA can quickly learn the system dynamics and perform close to the Whittle index policy, which matches well with our theoretical results.

\begin{figure}[t]
\centering
 \includegraphics[width=0.5\textwidth]{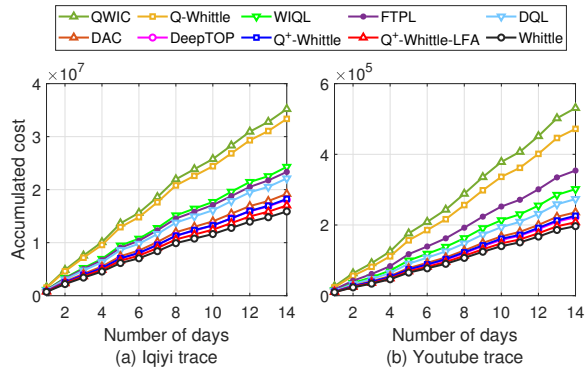}
\caption{Accumulated cost (latency) using real traces. }
\label{fig:learning-cost-real}
\vspace{-0.1in}
\end{figure}

\section{Conclusion}\label{sec:conclusion}
In this paper, we studied the content caching problem at the wireless edge with unreliable channels.  Our goal is to derive an optimal policy for making content caching decisions so as to minimize the average content request latency from end users.  We posed the problem in the form of a Markov decision process, and showed that the optimal policy has a simple threshold-structure and presented a closed form of Whittle indices for each content.  We then developed a novel model-free reinforcement learning algorithm with linear function approximation, which is called \qwhittleLFA that can fully exploit the structure of the optimal policy when the system parameters are unknown.  We mathematically characterized the performance of \qwhittleLFA and also numerically demonstrated its empirical performance.


\appendices

\section{Proof of Propositions in Section \ref{sec:indexable}}\label{Sec_Appendix_Prop}

\subsection{Proof of Proposition 1}\label{app:proof-prop-threshold}
\begin{proof}
According to Assumption~\ref{assumption1}, we denote the smallest state with no preference over active and passive actions as $R$, i.e, $Q^\alpha(R,1)= Q^\alpha(R,0)$. This implies the following two facts. First, for state $s<R$, the optimal action is $0$, i.e., 
\begin{align}\label{eq:fact1}
    J^\alpha(R-1)=R-1-W+\alpha J^\alpha(R).
\end{align}
Second, equal preference over two actions at state $R$ implies
\begin{align*}
   R-W+\alpha J^\alpha(R+1)
     &=R+\alpha P_R J^\alpha(R+1)\nonumber\allowdisplaybreaks\\
    &\qquad+\alpha(1-P_R)J^\alpha(R-1),
\end{align*}
from which we have 
\begin{align}\label{eq:W_equality}
    W=\alpha(1-P_R)(J^\alpha(R+1)-J^\alpha(R-1)).
\end{align}
From \eqref{eq:fact1}, we establish the connection between value functions of states $R-1$ and $R+1$, i.e.,
\begin{align}\label{eq:J_equality1}
    J^\alpha(R-1)=R\!-\!1\!-\!W\!+\!\alpha(R-W+\alpha J^\alpha(R+1)).
\end{align}
Substituting~(\ref{eq:J_equality1}) into \eqref{eq:W_equality}, we have
\begin{align*}
    J^\alpha(R+1)=\frac{\frac{W}{\alpha(1-P_R)}+(R-1-W+\alpha(R-W))}{1-\alpha^2}.
\end{align*}

As a result, $J^\alpha(R+1)$ can be updated as
\begin{align*}
  \begin{cases}
    R+1-W+\alpha J^\alpha(R+2),~ \text{if}~a=0,\\
    R+\!1\!+\alpha P_{R+1} J^\alpha(R+2)+\alpha(1-P_{R+1})J^\alpha(R),~ \text{if}~a=1.
    \end{cases}
\end{align*}

In the following, we show that it is optimal to choose action $1$ at state $R+1$.  We first show that $a=0$ is not optimal by contradiction, and then verify that $a=1$ is optimal. 
Assume that the optimal action at state $R+1$ is $a=0$. Then, we have
\begin{align}\label{eq:contra}
    W&\geq \alpha(1-P_{R+1})(J^\alpha(R+2)-J^\alpha(R))\nonumber\displaybreak[0]\\
    &=\alpha(1-P_{R+1})\Big(\frac{J^\alpha(R+1)-(R+1-W)}{\alpha}\nonumber\allowdisplaybreaks\\
    &\qquad\qquad-(R-W+\alpha J^\alpha(R+1))\Big)\nonumber\displaybreak[1]\\
    &= \frac{1-P_{R+1}}{\alpha(1-P_R)}W,
\end{align}
where the inequality is due to the fact that optimal action is $0$ at state $R+1$ and the last equality directly comes by plugging the closed-form expression of $J^\alpha(R+1)$. Since $\frac{1-P_{R+1}}{\alpha(1-P_R)}>1$, the inequality does not hold and it occurs an contradiction. This means that action $0$ is not optimal for state $R+1$.  We further verify that $a=1$ is optimal.  When the optimal action at state $R+1$ is $a=1$, we have
\begin{align}
    W&\stackrel{(a)}{\leq}\!\!\alpha(1\!-\!P_{R+1})\Bigg(\frac{J^\alpha(R+1)\!-\!(R+1-W)}{\alpha}\nonumber\allowdisplaybreaks\\
    &\qquad\qquad-(R-W+\alpha J^\alpha(R+1))\Bigg)\nonumber\displaybreak[0]\\
    &\stackrel{(b)}{\leq}\!\!\alpha(1\!-\!P_{R+1})\Bigg(\!\!\frac{J^\alpha(R\!+\!1)\!-\!(R\!+\!1)\!-\!\alpha(1\!\!-\!P_{R+1})J^\alpha(R)}{\alpha P_{R+1}}\nonumber\allowdisplaybreaks\\
    &\qquad\qquad-(R-W+\alpha J^\alpha(R+1))\Bigg)\nonumber\displaybreak[3]\\
    &= \alpha(1-P_{R+1})\Big(J^\alpha(R+2)-J^\alpha(R)\Big),
\end{align}
where (a) directly follows from the contradiction implied by \eqref{eq:contra}, and (b) holds as
  $\frac{J^\alpha(R+1)-(R+1)-\alpha(1-P_{R+1})J^\alpha(R)}{\alpha P_{R+1}}\geq \frac{J^\alpha(R+1)-(R+1-W)}{\alpha}. $ 
Thus the optimal action for $R+1$ is $1$. 

Following the same idea, the above results can be easily generalized to any state $s\geq R+1$, and hence we omit the detail here.  To this end, the optimal policy of the discounted MDP~\eqref{eq:discount_value} is of the threshold-type.

\end{proof}

\subsection{Proof of Proposition \ref{Prop:2}}
\begin{proof}
According to \cite{lippman1973semi},  the optimal expected total discounted latency $J_{\pi^*_\alpha}$ under the optimal policy $\pi_\alpha^*$ with discount factor $\alpha,$ and the optimal average latency $J_{\pi^*}$ under the optimal policy $\pi^*$ satisfy
   $\lim_{\alpha\rightarrow 1}(1-\alpha)J^\alpha_{\pi_\alpha^*}(s)=J_{\pi^*}(s), \forall s.$
Since our action set is finite, there exists an optimal stationary policy for the average latency problem such that $\pi_\alpha^*\rightarrow \pi^*$ \cite{lippman1973semi}. This shows that the optimal policy for \eqref{eq:decoupledduallag} is of the threshold-type. 
\end{proof}

\subsection{Proof of Proposition \ref{prop:indexable}}\label{app:proof-prop-indexable}

\begin{proof}
Since the optimal policy for~\eqref{eq:decoupledduallag} is of the threshold-type, for a given $W$, the optimal average cost under a threshold $R$ satisfies  
\begin{align}\label{eq:threshold_cost0}
h(W)\!:=\!\min_R\!\left\{{h^R(W)}\!:=\!\sum_{s=0}^{\infty}s\phi_{R}(s)\!-\!W\!\sum_{s=0}^R \phi_{R}(s)\right\},   
\end{align}
where $\phi_R(s)$ is the stationary probability of state $s$ under the threshold policy $\pi=R$.  It is easy to show that $h^R(W)$ is concave non-increasing in $W$ since it is a lower envelope of linear non-increasing functions in $W$, i.e., $h^R(W)>h^R(W^\prime)$ if $W<W^\prime.$  Thus we can choose a larger threshold $R$ when $W$ increases to further decrease the total cost according to \eqref{eq:threshold_cost0}, i.e, $D(W)\subseteq D(W^\prime)$ when $W<W^\prime$.
\end{proof}

\subsection{Proof of Proposition~\ref{prop:whittle-index-closed}}
\begin{proof}\label{proof:prop:whittle-index-closed}
Following from the definition of Whittle index, the performance of a policy with threshold $R$ equals to the performance of a policy with threshold $R+1$ \cite{larranaga2014index,larrnaaga2016dynamic}, i.e.,
\begin{align}
&\mathbb{E}_R[s] -{W(R)} \mathbb{E}_{R}[\mathds{1}_{\{ {A(s)=0}\}} ]\nonumber\allowdisplaybreaks\\
&
\qquad\qquad=\mathbb{E}_{R+1}[s] -W(R) \mathbb{E}_{R+1}[\mathds{1}_{\{ A(s)=0\}}],
\end{align}
where the subscript denotes the fact that the associated quantities involve a threshold policy with the value of threshold equal to this value. { Since the evolution of per-content is described by the transition kernel (a birth-and-death process) in  \eqref{eq:queuetrans}, we have $\mathbb{E}_R[ \mathds{1}_{\{ A(s)=0\}}]=\sum_{s=0}^R \phi_{R}(s)$.}
\end{proof}

\subsection{Proof of Proposition \ref{prop:stationary-distributions}}\label{app:proof-prop-distribution}
\begin{proof}
Given the transition kernel in \eqref{eq:queuetrans}, the transition rate satisfies $q(S+1|S,0)=q(S+1|S,1)=\lambda$ and $q(S_1|S,0)=0$ for $S\leq R$, and $q(S-1|S,1)=\nu S$ for $S>R$. It is clear that $\forall S<R$ is transient because the state keeps increasing. 
 Therefore, $\phi_R(S)=0,$ 
$\forall S<R$. Note that for threshold state $R$, the stationary probability satisfies
\begin{align*}
    \phi_R(R)=\frac{\nu (R+1)}{\lambda+\nu (R+1)} \phi_R(R+1).
\end{align*}
Based on the birth-and-death process, the stationary probabilities for states $R+l, \forall l=2, \cdots, S_{max}-R$ satisfy
\begin{align*}
   \phi_R(R+l)\frac{\lambda}{\lambda+\nu (R+l)}= \frac{\nu (R+l+1)}{\lambda+\nu (R+l+1)}\phi_R(R+l+1). 
\end{align*}
Therefore, we have the following relation
\begin{align*}
    \phi_R(R+l)&=\phi_{R}(R+1)\prod_{j=2}^l \frac{\lambda}{\lambda+\nu (R+j-1)}\frac{\lambda+\nu (R+j)}{\nu q(R+j)}.
\end{align*}
Since $\phi_R(R)+\phi_R(R+1)+\cdots+\phi_R(S_{max})=1$, we have  
 \begin{align*}
 \phi_R(R+1)&=1/\Bigg(1+\frac{\nu (R+1)}{\lambda+\nu (R+1)}\nonumber\allowdisplaybreaks\\
 &+\sum\limits_{l=2}^{S_{max}-R}\prod_{j=2}^l \frac{\lambda}{\lambda+\nu (R+j-1)}\frac{\lambda+\nu (R+j)}{\nu (R+j)}\Bigg).
 \end{align*}
\end{proof}

\section{Proof of Theorem \ref{thm:convergence}}\label{sec:QWhittle-learning-convergence-app}
To prove Theorem \ref{thm:convergence}, we need the following three key lemmas regarding the error terms defined in \eqref{eq:residual}. First, we study the property of $\tilde{\theta}_n$.
 
\begin{lemma}\label{lemma2}
Consider the iterates $\{{\theta}_n\}$ and $\{{W}_n\}$ generated by \eqref{eq:Q_2TSA}-\eqref{eq:W_2TSA}. Under 
Assumptions~\ref{assumption:Lipschitz}-\ref{assumption:stepsize}, we have for all $n\geq 0$,
\begin{align}\nonumber
     &\mathbb{E} \left[\|\tilde{\theta}_{n+1}\|^2|\mathcal{F}_n\right]\leq \gamma_n^2\Lambda+(1-2\gamma_n\mu_1+L_h^2\gamma_n^2)\|\tilde{{\theta}}_n\|^2\nonumber\displaybreak[0]\\
    &+2L_f^2L_g^2\eta_n^2\|\tilde{\theta}_n\|^2+2L_f^2L_g^2(L_f+1)^2\eta_n^2\|\tilde{W}_n\|^2\nonumber\displaybreak[1]\\
    &+\left(L_f^2\gamma_n^2+\frac{2(1+L_h\gamma_n)^2\eta_n^2L_g^2}{\gamma_n^2}\right)\|\tilde{{\theta}}_n\|^2\nonumber\allowdisplaybreaks\\
    &+\frac{2(1+L_h\gamma_n)^2\eta_n^2L_g^2(L_f+1)^2}{\gamma_n^2}\|\tilde{{W}}_n\|^2.\label{eq:Q_tilde}
\end{align}
\end{lemma}
\begin{proof}
According to the definition in \eqref{eq:residual}, we have 
\begin{align*}
    \tilde{\theta}_{n+1}&=\theta_{n+1}-f(W_{n+1})\nonumber\displaybreak[0]\\
    &=\tilde{\theta}_n+\gamma_nh(\theta_n,W_n)+\gamma_n\xi_n+f(W_n)-f(W_{n+1}),
\end{align*}
which leads to
\begin{align}
    &\|\tilde{\theta}_{n+1}\|^2=\|\tilde{\theta}_n\!+\!\gamma_nh(\theta_n,W_n)\!+\!\gamma_n\xi_n+f(W_n)\!-\!f(W_{n+1})\|^2\nonumber\displaybreak[0]\\
    &=\underset{\text{Term}_1}{\underbrace{\|\tilde{\theta}_n+\gamma_nh(\theta_n,W_n)\|^2}}+\underset{\text{Term}_2}{\underbrace{\|\gamma_n\xi_n+f(W_n)-f(W_{n+1})\|^2}}\nonumber\displaybreak[1]\\
    &\qquad+\underset{\text{Term}_3}{\underbrace{2\left(\tilde{\theta}_n+\gamma_n h(\theta_n,W_n)\right)^T\left(f(W_n)-f(W_{n+1})\right)}}\nonumber\displaybreak[2]\\
    &\qquad+\underset{\text{Term}_4}{\underbrace{2\gamma_n\left(\tilde{\theta}_n+\gamma_n h(\theta_n,W_n)\right)^T\xi_k}},
\end{align}
where the second equality is due to the fact that $\|\bold{x}+\bold{y}\|^2=\|\bold{x}\|^2+\|\bold{y}\|^2+2\bold{x}^T\bold{y}$.

We next analyze the conditional expectation of each term in $\|\tilde{\theta}_{n+1}\|^2$ on $\mathcal{F}_n$. We first focus on Term$_1$.

\begin{align*}
    &\mathbb{E}\Big[\text{Term}_1|\mathcal{F}_n\Big]\nonumber\allowdisplaybreaks\\
    &=\|\tilde{{\theta}}_k\|^2+2\gamma_n\tilde{{\theta}}_n^\intercal h({\theta}_n,{W}_n)+\|\gamma_nh({\theta}_n,{W}_n)\|^2\nonumber\displaybreak[1]\\
    &\overset{(a1)}{=}\!\!\|\tilde{{\theta}}_n\|^2\!\!\!+\!2\gamma_n\tilde{{\theta}}_n^\intercal h({\theta}_n,{W}_n\!)\!+\!\gamma_n^2\|h({\theta}_n,{W}_n\!)\!-\!h(f({W}_n),{W}_n)\|^2\nonumber\displaybreak[3]\\
    &\overset{(a2)}{\leq} \|\tilde{{\theta}}_n\|^2-2\gamma_n\mu_1\|\tilde{{\theta}}_n\|^2+L_h^2\gamma_n^2\|\tilde{{\theta}}_n\|^2,
\end{align*}
where (a1) follows from $h(f(W_n), W_n)=0$, and (a2) holds due to the Lipschitz continuity of $h$ in Assumption~\ref{assumption:Lipschitz} and $\gamma_n\tilde{\theta}_n^Th(\theta_n,W_n)\leq -\mu_1\|\tilde{\theta}_n\|^2$.  For Term$_2$, we have

\begin{align}
    &\mathbb{E}\Big[\text{Term}_2|\mathcal{F}_n\Big]\nonumber\allowdisplaybreaks\\
    &=\mathbb{E}[\|f(W_n)-f(W_{n+1})+\gamma_n\xi_n\|^2|\mathcal{F}_n]\nonumber\\
    &\overset{(b1)}{=}\mathbb{E}[\|f(W_n)-f(W_{n+1})\|^2|\mathcal{F}_n]+\gamma_n^2\mathbb{E}[\|\xi_n\|^2|\mathcal{F}_n]\nonumber\\
    &\overset{(b2)}{\leq} L_f^2\mathbb{E}[\|W_n-W_{n+1}\|^2|\mathcal{F}_n]+\gamma_n^2\Lambda\nonumber\\
    &=L_f^2\mathbb{E}[\|\eta_n g(\theta_n,W_n)\|^2|\mathcal{F}_n]+\gamma_n^2\Lambda\nonumber\\
    &=L_f^2\eta_n^2\|g(\theta_n,W_n)\|^2+\gamma_n^2\Lambda\nonumber\\
    &\overset{(b3)}{\leq} 2L_f^2\eta_n^2\|g(\theta_n,W_n)-g(f(W_n),W_n)\|^2+\gamma_n^2\Lambda \nonumber\displaybreak[0]\\
    &\qquad+2L_f^2\eta_n^2\|g(f(W_n),W_n)-g(f(W(R)),W(R))\|^2\nonumber\displaybreak[1]\\
    &\overset{(b4)}{\leq} 2L_g^2L_f^2\eta_n^2\|\tilde{\theta}_n\|^2+2L_g^2L_f^2\eta_n^2\Big(\|f(W_n)-f(W(R))\|\nonumber\displaybreak[2]\\
    &\quad+\|W_n-W(R)\|\Big)^2+\gamma_n^2\Lambda\nonumber\displaybreak[3]\\
    &\overset{(b5)}{\leq} 2L_f^2L_g^2\eta_n^2\|\tilde{\theta}_n\|^2\!\!\!+\!2L_f^2L_g^2(L_f\!+\!1)^2\eta_n^2\|\tilde{W}_n\|^2\!+\!\gamma_n^2\Lambda,
\end{align}
where (b1) is due to $\mathbb{E}[\xi_n|\mathcal{F}_n]=0,$ (b2) is due to the Lipschitz continuity of $f$,
and (b3) holds since $ \|g(\theta_n,W_n)\|^2\leq 2\|g(\theta_n,W_n)-g(f(W_n),W_n)\|^2+2\|g(f(W_n),W_n)-g(f(W(R)),W(R))\|^2$ when $g(f(W_n), W_n)=0,$ (b4) and (b5) hold because of the Lipschitz continuity of $g$ and $f$.
Next, we have the conditional expectation of Term$_3$ as
\begin{align}
&\mathbb{E}\Big[\text{Term}_3|\mathcal{F}_n\Big]\nonumber\allowdisplaybreaks\\
   &\leq2\mathbb{E}\|\tilde{\theta}_n+\gamma_nh(\theta_n,\!W_n)\|\cdot\|f(W_n)-f(W_{n+1})\|\nonumber\displaybreak[0]\\
    &\overset{(c1)}{\leq} 2L_f\eta_n\|\tilde{\theta}_n+\gamma_nh(\theta_n,W_n)\|\cdot\|g(\theta_n, W_n))\|\nonumber\displaybreak[1]\\
    &\leq 2L_f\eta_n(1+L_h\gamma_n)\|\tilde{\theta}_n\|\left(L_g\|\tilde{\theta}_n\|+L_g(L_f+1)\|\tilde{W}_n\|\right)\nonumber\displaybreak[2]\\
     &\overset{(c2)}{\leq} L_f^2\gamma_n^2\|\tilde{{\theta}}_n\|^2+ \frac{(1+L_h\gamma_n)^2\eta_n^2}{\gamma_n^2}\nonumber\allowdisplaybreaks\\
     &\qquad\qquad\cdot\left(L_g\|\tilde{{\theta}}_n\|^2+L_g(L_f+1)\|\tilde{{W}}_n\|\right)^2\nonumber\displaybreak[0]\\
    &\leq \left(L_f^2\gamma_n^2+\frac{2(1+L_h\gamma_n)^2\eta_n^2L_g^2}{\gamma_n^2}\right)\|\tilde{{\theta}}_n\|^2\nonumber\allowdisplaybreaks\\
    &\qquad+\frac{2(1+L_h\gamma_n)^2\eta_n^2L_g^2(L_f+1)^2}{\gamma_n^2}\|\tilde{{W}}_n\|^2,
\end{align}
where (c1) is due to the Lipschitz continuity of $f$ and (c2) holds because $2\bold{x}^T\bold{y}\leq \beta\|\bold{x}\|^2+1/\beta\|\bold{y}\|^2, \forall \beta>0$.
Since $\mathbb{E}\Big[\text{Term}_4|\mathcal{F}_n\Big]=0$, combining all terms leads to the final expression in \eqref{eq:Q_tilde}.
\end{proof}

\begin{lemma}\label{lemma3}
Consider the iterates $\{{\theta}_n\}$ and $\{{W}_n\}$ generated by \eqref{eq:Q_2TSA}-\eqref{eq:W_2TSA}. Under 
Assumptions~\ref{assumption:Lipschitz}-\ref{assumption:stepsize}, for any $n\geq 0$, we have 
\begin{align}
  \mathbb{E} &\left[\|\tilde{W}_{n+1}\|^2|\mathcal{F}_n\right]{\leq} \|\tilde{W}_{n}\|^2+2\eta_n^2L_g^2\|\tilde{Q}_n\|^2\nonumber\allowdisplaybreaks\\
  &\qquad\qquad+2\eta_n^2L_g^2(L_h+1)^2\|\tilde{W}_n\|^2. \label{eq:w_tilde}
\end{align}
\end{lemma}

\begin{proof}
According to \eqref{eq:residual}, we have 
    $\tilde{W}_{n+1}=W_{n+1}-W(R)=\tilde{W}_n+\eta_ng(\theta_n,W_n),$
which leads to 
\begin{align}
    &\mathbb{E}\left[\|\tilde{W}_{n+1}\|^2|\mathcal{F}_n\right]\nonumber\allowdisplaybreaks\\
    &=\|\tilde{W}_n\|^2+2\eta_n\tilde{W}_n^Tg(\theta_n, W_n)+\eta_n^2\|g(\theta_n,W_n)\|^2\nonumber\displaybreak[1]\\
    &\overset{(d1)}{\leq} \|\tilde{W}_n\|^2-2\eta_n\mu_2\|\tilde{W}_n\|^2+\eta_n^2\|g(\theta_n,W_n)|^2\nonumber\displaybreak[2]\\
    &\overset{(d2)}{\leq} \|\tilde{W}_{n}\|^2-2\eta_n\mu_2\|\tilde{W}_n\|^2\nonumber\allowdisplaybreaks\\
    &\qquad+2\eta_n^2L_g^2\|\tilde{\theta}_n\|^2+2\eta_n^2L_g^2(L_f+1)^2\|\tilde{W}_n\|^2,
\end{align}
where (d1) is due to $2\eta_n\tilde{W}_n^T(g(\theta_n, W_n))\leq -2\mu_2\|\tilde{W}_n\|^2$ and { (d2) is due to (b3)-(b5).} 
\end{proof}

\begin{lemma}\label{lemma4}
Consider the iterates $\{{\theta}_n\}$ and $\{{W}_n\}$ generated by \eqref{eq:Q_2TSA}-\eqref{eq:W_2TSA}. 
Assume that $\gamma_n\leq\min\left( \frac{1}{2\mu_1},\frac{2\mu_1}{L_h^2+L_f^2}\right)$, $\eta_n\leq \min\left( \frac{1}{2\mu_2}, \frac{\mu_2}{L_g^2(L_f+1)^2(L_f^2+1)}\right)$ and $\eta_n\ll\gamma_n$. Then under Assumptions~\ref{assumption:Lipschitz}-\ref{assumption:stepsize}, 
we have 
\begin{align*}
    \lim_{n\rightarrow\infty}\mathbb{E}\left[\|\tilde{{\theta}}_n\|^2+\|\tilde{{W}}_n\|^2|\mathcal{F}_n\right]\rightarrow 0~ almost~ surely.
\end{align*}

\end{lemma}
\begin{proof}
Providing Lemma~\ref{lemma2} and Lemma~\ref{lemma3},  we have
\begin{align*}
    &\mathbb{E}\left[\|\tilde{{\theta}}_n\|^2+\|\tilde{{W}}_n\Big\|^2|\mathcal{F}_n\right]\nonumber\allowdisplaybreaks\\
    &\leq\gamma_n^2\Lambda+(1-2\gamma_n\mu_1+L_h^2\gamma_n^2)\|\tilde{{\theta}}_n\|^2\nonumber\allowdisplaybreaks\\
    &+\Big(2L_f^2L_g^2\eta_n^2\|\tilde{{\theta}}_n\|^2+2L_f^2L_g^2\eta_n^2(L_f+1)^2\|\tilde{{W}}_n\|^2\Big)\nonumber\displaybreak[0]\\
    &+\left(L_f^2\gamma_n^2+\frac{2(1+L_h\gamma_n)^2\eta_n^2L_g^2}{\gamma_n^2}\right)\|\tilde{{\theta}}_n\|^2\nonumber\allowdisplaybreaks\\
    &+\frac{2(1+L_h\gamma_n)^2\eta_n^2L_g^2(L_f+1)^2}{\gamma_n^2}\|\tilde{{W}}_n\|^2\nonumber\displaybreak[1]\\
    &+\!(1\!-\!2\eta_n\mu_2)\|\tilde{{W}}_{n}\|^2\!\!\!+\!2\eta_n^2L_g^2\|\tilde{{\theta}}_n\|^2\!\!\!+\!2\eta_n^2L_g^2(L_f\!+\!1)^2\|\tilde{{W}}_n\|^2\nonumber\displaybreak[2]\\
    & \leq\gamma_n^2\Lambda+(1-2\gamma_n\mu_1)\|\tilde{{\theta}}_n\|^2+(1-2\eta_n\mu_2)\|\tilde{{W}}_{n}\|^2\displaybreak[3]\\
    &+\!\! \Bigg(\!\!L_h^2\gamma_n^2\!+\!2L_f^2L_g^2\eta_n^2\!+\!L_f^2\gamma_n^2\!+\!\frac{2(1+L_h\gamma_n)^2\eta_n^2L_g^2}{\gamma_n^2}\nonumber\allowdisplaybreaks\\
    &\qquad\qquad\qquad+\!2\eta_n^2L_g^2\Bigg)\|\tilde{{\theta}}_n\|^2\\
    &+2L_g^2(L_f+1)^2\left(L_f^2\eta_n^2+\eta_n^2+\frac{(1+L_h\gamma_n)^2\eta_n^2}{\gamma_n^2}\right)\|\tilde{{W}}_{n}\|^2.
\end{align*}
Since $\eta_n\leq \min\left( \frac{1}{2\mu_2}, \frac{\mu_2}{L_g^2(L_f+1)^2(L_f^2+1)}\right)$, $\gamma_n\leq\min\left( \frac{1}{2\mu_1},\frac{2\mu_1}{L_h^2+L_f^2}\right)$  and $\eta_n\ll\gamma_n$, we have
\begin{align*}
    -1\leq D_1:&=-2\gamma_n\mu_1+\Bigg(L_h^2\gamma_n^2+2L_f^2L_g^2\eta_n^2+L_f^2\gamma_n^2\allowdisplaybreaks\\
    &+\frac{2(1+L_h\gamma_n)^2\eta_n^2L_g^2}{\gamma_n^2}+2\eta_n^2L_g^2\Bigg)\leq 0,\\
    -1\leq D_2:&=-2\eta_n\mu_2+2L_g^2(L_f+1)^2\allowdisplaybreaks\\
    &\cdot\left(L_f^2\eta_n^2+\eta_n^2+\frac{(1+L_h\gamma_n)^2\eta_n^2}{\gamma_n^2}\right)\leq 0.
\end{align*}
Define $x_n=\min(D_1, D_2)$. Then, we have
\begin{align*}
   &\mathbb{E}\left[\|\tilde{{\theta}}_n\|^2+\|\tilde{{W}}_n|^2|\mathcal{F}_n\right] \nonumber\allowdisplaybreaks\\
   &\leq \gamma_n^2\Lambda+ (1+x_n)\left(\|\tilde{{\theta}}_n\|^2+\|\tilde{{W}}_n\|^2\right)\\
    &= \prod_{t=0}^n(1+x_t)\left(\|\tilde{{\theta}}_0\|^2+\|\tilde{{\lambda}}_0\|^2\right)\nonumber\allowdisplaybreaks\\
    &\qquad+\left(1+\sum_{t=0}^n\prod_{\tau=0}^t(1+x_{n-\tau})\right)\gamma_n^2\Lambda.
\end{align*}
Since $0\leq 1+x_n\leq 1, \forall n$ and $\lim_{n\rightarrow\infty}\gamma_n\rightarrow 0$, $\lim_{n\rightarrow\infty}\mathbb{E}\left[\|\tilde{{\theta}}_n\|^2+\|\tilde{W}_n\|^2|\mathcal{F}_n\right]\rightarrow 0$ \emph{almost surely}.
 
\end{proof}

Now we are ready to prove Theorem \ref{thm:convergence}.  Providing Lemmas~\ref{lemma2}-\ref{lemma4}, if $\frac{\eta_{n}}{\alpha_n}$ is non-increasing, we have 

\begin{align}\nonumber
     &\mathbb{E}\Big[M({\theta}_{n+1},{W}_{n+1})\Big|\mathcal{F}_n\Big]\nonumber\allowdisplaybreaks\\
    &\leq \frac{\eta_n}{\gamma_n}\gamma_n^2\Lambda+\frac{\eta_n}{\gamma_n}(1-2\gamma_n\mu_1+L_h^2\gamma_n^2)\|\tilde{{\theta}}_n\|^2\nonumber\\
    &+\frac{\eta_n}{\gamma_n}\Big(2L_f^2L_g^2\eta_n^2\|\tilde{{\theta}}_n\|^2+2L_f^2L_g^2\eta_n^2(L_f+1)^2\|\tilde{{W}}_n\|^2\Big)\nonumber\\
    &\qquad+\frac{\eta_n}{\gamma_n}\left(L_f^2\gamma_n^2+\frac{2(1+L_h\gamma_n)^2\eta_n^2L_g^2}{\gamma_n^2}\right)\|\tilde{{\theta}}_n\|^2\nonumber\allowdisplaybreaks\\
    &\qquad+\frac{2(1+L_h\gamma_n)^2\eta_n^3L_g^2(L_f+1)^2}{\gamma_n^3}\|\tilde{{W}}_n\|^2\nonumber\allowdisplaybreaks\\
    &\qquad+(1-2\eta_n\mu_2)\|\tilde{{W}}_{n}\|^2+2\eta_n^2L_g^2\|\tilde{{\theta}}_n\|^2\nonumber\allowdisplaybreaks\\
    &\qquad+2\eta_n^2L_g^2(L_f+1)^2\|\tilde{{W}}_n\|^2 \nonumber\\
     & \overset{(e1)}{\leq}\frac{\eta_n}{\gamma_n}\gamma_n^2\Lambda+\frac{\eta_n}{\gamma_n}(1-2\gamma_n\mu_1)\|\tilde{{\theta}}_n\|^2+(1-2\eta_n\mu_2)\|\tilde{{W}}_{n}\|^2\nonumber\displaybreak[0]\\
    &+ \Bigg((L_h^2+L_f^2)\gamma_n\eta_n+2(L_f^2+1)L_g^2\frac{\eta_n^3}{\gamma_n}\nonumber\allowdisplaybreaks\\
    &\qquad\qquad\qquad\qquad+\frac{2(1+L_h\gamma_n)^2\eta_n^3L_g^2}{\gamma_n^3}\Bigg)\|\tilde{{\theta}}_n\|^2 \nonumber\\
    &+2L_g^2(L_f+1)^2\Bigg(L_f^2\frac{\eta_n^3}{\gamma_n}+\gamma_n\eta_n+\frac{(1+L_h\gamma_n)^2\eta_n^3}{\gamma_n^3}\Bigg)\|\tilde{{W}}_{n}\|^2\nonumber\displaybreak[1]\\
     &\overset{(e2)}{ \leq}\max(1-2\gamma_n\mu_1, 1-2\eta_n\mu_2)\mathbb{E}\Big[M({\theta}_{n+1},{W}_{n+1})\Big|\mathcal{F}_n\Big]\nonumber\displaybreak[2]\\
    &+\!\frac{\eta_n}{\gamma_n}\gamma_n^2\Lambda\!+\! (L_h^2\!+\!L_f^2\!+\!2L_g^2(L_f\!+\!1)^2)\gamma_n\eta_n\!\!\left(\|\tilde{{\theta}}_n\|^2\!\!\!+\!\|\tilde{{W}}_n\|^2\!\right)\nonumber\displaybreak[3]\\
    &+\!\!2L_g^2(L_f\!\!+\!\!1)^2\!\left(L_f^2\!+\!{(1\!\!+\!\!L_h\gamma_n)^2}\right)\frac{\eta_n^3}{\gamma_n^3}\!\!\left(\!\|\tilde{{\theta}}_{n}\|^2\!+\!\|\tilde{{W}}_{n}\|^2\!\right).\label{eq:MSE}
\end{align}
Since $(n+1)^2\cdot\gamma_n\eta_n=\gamma_0\eta_0(n+1)^{1/3}$ and $(n+1)^2\cdot\frac{\eta_n^3}{\gamma_n^2}=\frac{\eta_0^2}{\gamma^2_0}(n+1)^{1/3}$,
multiplying both sides of (\ref{eq:MSE}) with $(n+1)^2$, we have 
\begin{align}\nonumber
    &(n+1)^2\mathbb{E}\Big[M({\theta}_{n+1},W_{n+1})\Big|\mathcal{F}_n\Big]\nonumber\allowdisplaybreaks\\
    &{\leq} (n+1)^{1/3}\gamma_0\eta_0\Lambda+n^2\mathbb{E}\Big[M({\theta}_{n},W_{n})\Big|\mathcal{F}_n\Big]\nonumber\displaybreak[0]\\
    &+\!(L_h^2\!\!+\!\!L_f^2\!\!+\!\!2L_g^2(L_f\!+\!1)^2)\alpha_0\eta_0(n\!+\!1)^{1/3}\!\left(\|\tilde{{\theta}}_n\|^2\!\!+\|\tilde{{W}}_n\|^2\!\right)\nonumber\displaybreak[1]\\
    &+\!2L_g^2(L_f\!\!+\!\!1)^2\!\!\left(L_f^2\!\!+\!\!{(1\!\!+\!L_h\gamma_n)^2}\right)\frac{\eta_0^3}{\gamma_0^3}(n\!\!+\!\!1)^{\frac{1}{3}}\!\!\left(\|\tilde{{\theta}}_{n}\|^2\!\!+\!\|\tilde{{W}}_{n}\|^2\!\right)\nonumber\displaybreak[2]\\
    &\leq(n+1)^{1/3}\gamma_0\eta_0\Lambda+ n^2\mathbb{E}\Big[M({\theta}_{n},W_{n})\Big|\mathcal{F}_n\Big]\nonumber\allowdisplaybreaks\\
    &\qquad+(n+1)^{1/3}\left(C_1\left(\|\tilde{{\theta}}_0\|^2+|\tilde{W}_0\|^2\right)\right), \label{eq: recursion}
\end{align}
where $C_1=(L_h^2+L_f^2+2L_g^2(L_f+1)^2)\alpha_0\eta_0+2L_g^2(L_f+1)^2\left(L_f^2+{(1+L_h\alpha_0)^2}\right)\frac{\eta_0^3}{\alpha_0^3}$.
Summing \eqref{eq: recursion} from time step $0$ to time step $n$, we have 
\begin{align}
    &(n+1)^2\mathbb{E}\Big[M({\theta}_{n+1},W_{n+1})\Big|\mathcal{F}_n\Big]
   \nonumber\allowdisplaybreaks\\
   \leq& \mathbb{E}\Big[M({\theta}_0,W_0)\Big]+(n+1)^{4/3}C_1(\|\tilde{\theta}_0\|^2+\|\tilde{W}_0\|^2)\nonumber\allowdisplaybreaks\\
   &\qquad\qquad+(n+1)^{4/3}\gamma_0\eta_0\Lambda. \label{eq:last_eqn}
\end{align}
Finally, dividing both sides by $(n+1)^2$ yields the results in Theorem \ref{thm:convergence}.

\bibliographystyle{IEEEtran}
\bibliography{refs}

\end{document}